\documentclass[11pt,a4paper]{article}
\usepackage[margin=1in]{geometry}
\usepackage{lmodern}
\usepackage[T1]{fontenc}
\usepackage[utf8]{inputenc}
\usepackage[tbtags]{amsmath}
\allowdisplaybreaks
\usepackage{amssymb,amsthm}
\usepackage{hyperref}
\usepackage[svgnames]{xcolor}
\usepackage[capitalise,nameinlink]{cleveref}
\hypersetup{colorlinks={true},linkcolor={DarkBlue},citecolor=[named]{DarkGreen}}
\usepackage[numbers,sort&compress]{natbib}
\usepackage{doi}
\usepackage{tikz} 
\usepackage{pgfplots}
\pgfplotsset{compat=1.14}
\usepgfplotslibrary{fillbetween}
\usepackage{hyphenat}
\usepackage[font=footnotesize]{caption}
\usepackage{multirow}

\newtheorem{theorem}{Theorem}
\newtheorem{lemma}{Lemma}

\newtheorem{corollary}{Corollary}

\theoremstyle{definition}
\newtheorem{definition}{Definition}
\newtheorem{claim}{Claim}
\newtheorem{remark}{Remark}

\newcommand{\mech}{\ensuremath{\mathcal{M}}}
\newcommand{\mechsp}{\ensuremath{\mathcal{SP}}}
\newcommand{\poa}{\mathrm{PoA}}
\newcommand{\pos}{\mathrm{PoS}}
\newcommand{\then}{\Longrightarrow}
\newcommand{\vecc}[1]{\ensuremath{\mathbf{#1}}}
\newcommand{\sset}[1]{\left\{ #1\right\}}
\newcommand{\ssets}[1]{\{ #1\}}
\newcommand{\fwh}[1]{\; \left| \; #1 \right.}
\newcommand{\fwhs}[1]{\; | \; #1 }
\newcommand{\opt}{\ensuremath{\mathrm{OPT}}}
\DeclareMathOperator*{\argmin}{argmin}
\newcommand{\inters}{\cap} 
\newcommand{\map}{\longrightarrow}
\renewcommand{\t}{\vecc{t}}
\renewcommand{\ss}{\vecc{s}}
\renewcommand{\AA}{\mathcal{A}}
\def \R{\mathbb R}
\def \Rgeq {{\R}_{\geq 0}}
\def \N{\mathbb N}

\author{
Aris Filos-Ratsikas\thanks{University of Liverpool. Email: \href{mailto:Aris.Filos-Ratsikas@liverpool.ac.uk}{\nolinkurl{Aris.Filos-Ratsikas@liverpool.ac.uk}}}
\and
Yiannis Giannakopoulos\thanks{TU Munich. Email: \href{mailto:yiannis.giannakopoulos@tum.de}{\nolinkurl{yiannis.giannakopoulos@tum.de}}}
\and
Philip Lazos\thanks{Sapienza University of Rome. Email: 
\href{mailto:plazos@gmail.com}{\nolinkurl{lazos@diag.uniroma1.it}}}
}

\title{The Pareto Frontier of Inefficiency in Mechanism Design\thanks{Supported by
ERC Advanced Grant 321171 (ALGAME), the Swiss National Science Foundation under
contract No.~200021\_165522 and the Alexander von Humboldt Foundation with funds
from the German Federal Ministry of Education and Research (BMBF). Also partially
supported by the ERC Advanced Grant 788893 AMDROMA ``Algorithmic and Mechanism
Design Research in Online Markets'' and MIUR PRIN project ALGADIMAR ``Algorithms,
Games, and Digital Markets''.
Y.\ Giannakopoulos is an associated researcher with the Research Training Group GRK 2201 ``Advanced Optimization in a Networked Economy'', funded by the German Research Foundation (DFG).
\newline\indent An extended abstract of this paper appeared in WINE'19~\cite{fgl2019}.
}}

\date{February 18, 2021}

\begin{document}

\maketitle
\begin{abstract}
	We study the trade-off between the Price of Anarchy (PoA) and the Price of
		Stability (PoS) in mechanism design, in the prototypical problem of
		unrelated machine scheduling. We give bounds on the space of feasible
		mechanisms with respect to the above metrics, and observe that two
		fundamental mechanisms, namely the First-Price (FP) and the Second-Price
		(SP), lie on the two opposite extrema of this boundary. Furthermore, for the
		natural class of anonymous task-independent mechanisms, we completely
		characterize the PoA/PoS Pareto frontier; we design a class of optimal
		mechanisms $\mechsp_\alpha$ that lie \emph{exactly} on this frontier. In
		particular, these mechanisms range smoothly, with respect to parameter
		$\alpha\geq 1$ across the frontier, between the First-Price ($\mechsp_1$)
		and Second-Price ($\mechsp_\infty$) mechanisms.

		En route to these results, we also provide a definitive answer to an
		important question related to the scheduling problem, namely whether
		non-truthful mechanisms can provide better makespan guarantees in the
		equilibrium, compared to truthful ones. We answer this question in the
		negative, by proving that the Price of Anarchy of \emph{all} scheduling
		mechanisms is at least $n$, where $n$ is the number of machines.
\end{abstract}

\section{Introduction}
\label{sec:intro}

The field of \emph{algorithmic mechanism design} was established in the seminal
paper of~\citet{Nisan:2001aa} and has ever since been at the centre
of research in the intersection of economics and computer science. The research
agenda put forward in \cite{Nisan:2001aa} advocates the study of approximate
solutions to interesting optimization problems, in settings where rational agents
are in control of the input parameters. More concretely, the authors of
\cite{Nisan:2001aa} proposed a framework in which, not unlike classical approaches
in approximation algorithms, algorithms that operate under certain limitations are
evaluated in terms of their approximation ratio. In particular, in algorithmic
mechanism design, this constraint comes from the requirement that agents should have
the right incentives to always report their inputs \emph{truthfully}. The
corresponding algorithms, paired with appropriately chosen payment functions, are
called \emph{mechanisms} \cite{Nisan07}.

Another pioneering line of work, initiated by~\citet{elias} and popularized further
by~\citet{Roughgarden2002a}, studies the \emph{inefficiency} of games through the
notion of the \emph{Price of Anarchy (PoA)}, which measures the deterioration of
some objective at the worst-case Nash equilibrium. A more optimistic version of the
same principle, where the inefficiency is measured at the \emph{best}
equilibrium~\citep{Schulz2003}, was introduced in~\cite{anshelevich2008price}, under
the name of \emph{Price of Stability (PoS)}.

Given the straightforward observation that mechanisms induce games between the
agents that control their inputs, as well as the fact that truthfulness is typically
a very demanding property, an alternative approach to the framework
of~\citet{Nisan:2001aa} is to design mechanisms that perform well \emph{in the
equilibrium}, i.e., they provide good PoA or PoS guarantees. This approach has been
adopted, among others, by central papers in the field (e.g., see
\cite{syrgkanis2013composable,roughgarden2017price} and references therein) and is
by now as much a part of algorithmic mechanism design as the original framework of
\citep{Nisan:2001aa}. An interesting question that has arisen in many settings is
whether non-truthful mechanisms (evaluated at the worst-case equilibrium, in terms
of their PoA) can actually outperform truthful ones (evaluated at the truth-telling,
dominant strategy equilibrium), for a given objective
\cite{K14,GKK16,christodoulou2016social}.

While the literature that studies the concepts of PoA and PoS is long and extensive,
there seems to be a lack of a \emph{systematic approach} investigating the trade-off
between the two notions \emph{simultaneously}. More concretely, given a problem in
algorithmic mechanism design, it seems quite natural to explore not only the best
mechanisms in terms of the two notions independently, but also the mechanisms that
achieve the best trade-off between the two. In a sense, this approach concerns a
``tighter'' optimality notion, as among a set of mechanisms with an ``acceptable''
Price of Anarchy guarantee, we would like to identify the ones that provide the best
possible Price of Stability. Our main contribution in the current paper is the proposal
of such a research agenda and its application on the canonical problem in the field,
introduced in the seminal work of~\citet{Nisan:2001aa}, that of scheduling on
unrelated machines.

\subsection{Our Contributions}
\label{sec:intro_results}

\paragraph{PoA/PoS trade-off:} We propose the \emph{research agenda of studying systematically the
trade-off between the Price of Anarchy and the Price of Stability in algorithmic
mechanism design}. Specifically, given a problem at hand and an objective function,
we are interested in the trade-off between the PoA and the PoS
of mechanisms for the given objective.
We apply this approach on the prototypical problem of algorithmic mechanism design
studied in \cite{Nisan:2001aa}, that of unrelated machine scheduling, where the
machines are self-interested agents.

First, in~\cref{sec:pareto_all}, for the class of \emph{all} possible mechanisms,
we prove that PoA guarantees imply corresponding PoS lower bounds and vice-versa
(\cref{thm:tradeoff}), which allows us to quantify the possible trade-off between
the two inefficiency notions in terms of a feasible region
(see~\cref{fig:pareto_general}); we refer to the boundary of this region as the
\emph{inefficiency boundary}. Interestingly, two well-known
mechanisms, namely the First-Price and the Second-Price mechanisms, turn out to lie on the
extreme points of this boundary.

Next, in~\cref{sec:task-independent}, for the well-studied class of task-independent
and anonymous mechanisms,\footnote{We remark that the best known mechanisms for
several variants of truthful scheduling are task-independent and anonymous. In
\cref{sec:discussion}, we provide a more detailed discussion, as well as an almost
matching trade-off bound for mechanisms that need not be anonymous
(see~\cref{thm:sqrt-tradeoff}).} we are able to show a tighter feasibility region
(\cref{thm:task_independent_tradeoff}). As a matter of fact, its inefficiency
boundary turns out to \emph{completely characterize} the achievable trade-off
between the PoA and the PoS: we design a class of mechanisms (\cref{sec:spa_mechs})
called $\mechsp_\alpha$, parameterized by a quantity $\alpha$, which are
\emph{optimal} in the sense that for any possible trade-off between the two
inefficiency notions, there exists a mechanism in the class (i.e., an appropriate
choice of $\alpha$) that exactly achieves this trade-off
(\cref{thm:sp_alpha_poa,thm:sp_alpha_pos}). In other words, we obtain an exact
description of the \emph{Pareto frontier of inefficiency}
(see~\cref{fig:pareto_task_indi}).

Our $\mechsp_\alpha$ mechanisms are simple and
intuitive and are based on the idea of setting reserve prices \emph{relatively} to
the declarations of the fastest machines. While this is clearly not truthful, we
prove that it induces the equilibria which are desirable for our results. More
precisely, the choice of $\alpha$ enables us to ``control'' the set of possible
equilibria in a way that allows us to achieve any trade-off on the boundary.

\paragraph{The Price of Anarchy of scheduling:} Our results also offer insights in
an other interesting direction. The inefficiency boundary result for general
mechanisms is based on a novel monotonicity lemma (\cref{lem:monotonicity}), which
is quite different from the well-known \emph{weak monotonicity}
property~\citep{saks2005weak} (see, e.g.,
\cite{Nisan:2001aa,christodoulou2009lower}). Interestingly, we also use this lemma
to prove a \emph{general lower bound of $n$} on the PoA of \emph{any} mechanism for
the scheduling problem (\cref{thm:lower}), where $n$ is the number of machines. This
result contributes to the intriguing
debate~\citep{K14,GKK16,christodoulou2016social} of whether general mechanisms (that
may be non-truthful, evaluated at the worst-case equilibrium) can outperform
truthful ones (evaluated at the truth-telling equilibrium). Given that the best
known truthful mechanism achieves an $n$-approximation, our results here provide a
definitive, negative answer to the aforementioned question (see
\cref{sec:truthful-vs-strategic} for a more detailed discussion). As a matter of
fact, in \cref{thm:PoAtruthful}, we actually show that when evaluated at their
worst-case equilibrium, truthful mechanisms are bound to perform even more poorly,
as their PoA is unbounded.  \medskip \medskip

Finally, in \cref{sec:discussion}, we conclude with a detailed discussion, where we
identify several intriguing directions for future work, both on a technical and a
conceptual level.

\subsection{Related Work}
\label{sec:intro_related}

\subsubsection{The Algorithmic Scheduling Problem}

The algorithmic version of the scheduling problem (without any consideration to
incentives) is one of the most fundamental problems in computer science, whose
origins can be traced back to the works of \citet{johnson1954optimal},
\citet{jackson1955scheduling} and \citet{graham1966bounds}. The problem is often
also generally referred to as the ``Job Shop Scheduling Problem''
\cite{garey1976complexity}, as it accurately models job assignment problems in
manufacturing systems. On top of this connection, the machine scheduling problem in
fact enjoys a plethora of applications, ranging from classical problems in
distributed computing, such as assigning computational tasks to parallel processors,
to newer applications in multi-agent systems, such as assigning vehicles to charging
stations. For more details and applications, we refer the reader to some of the
works on the algorithmic version~\citep{ibarra1977heuristic,Davis1981,LST90}, as
well as the surveys of \citet{HochbaumHall97}, \citet{potts2009fifty} and
\citet{lenstra1977complexity}, and the books of \citet{pinedo2012scheduling} and
\citet{kan2012machine}.

\subsubsection{The Selfish Scheduling Problem}

The scheduling problem on unrelated selfish machines is the prototypical problem
studied by~\citet{Nisan:2001aa} in 1999, when they introduced the field of
algorithmic mechanism design. The authors consider the worst-case performance of
truthful mechanisms on dominant strategy, truth-telling equilibria, and discover
that the well-known Second-Price auction\footnote{In the related literature, this
mechanism is often referred to as the Vickrey-Clarke-Groves (VCG)
mechanism~\citep{Vickrey1961a,Clarke1971a,Groves1973a}. The mechanism was originally
referred to as the ``minWork Mechanism'' in~\citep{Nisan:2001aa}.} has an
approximation ratio of $n$ for the problem, where $n$ is the number of machines.
Despite several attempts over the years, this is still the best-known truthful
mechanism. On the other hand, the succession of the best proven lower bounds started
with $2$ in \cite{Nisan:2001aa}, improved to $2.41$
by~\citet{christodoulou2009lower} and to $2.61$ by~\citet{KV13}, and
finally\footnote{During the preparation of our paper, a new manuscript
by~\citet{dobzinski2020improved} appeared online, further improving the lower bound
to $2.80$.} to $2.75$ in the very recent work of~\citet{giannakopoulos2020new}.
Interestingly, \citet{ashlagi2012optimal}
showed a matching lower bound of $n$ for \emph{anonymous} mechanisms (i.e.,
mechanisms that do not take the identities of the machines into account) and whether
there is a better mechanism that is not anonymous is still the most prominent open
problem in the area. In any case, anonymity is in general a desirable property which
is satisfied by most natural mechanisms (including the best known mechanisms for
scheduling~\cite{LST90}); we further discuss the role of this property in our
setting in~\cref{sec:discussion}.

Several other variants of the problem have also been considered over the years, such
as randomized mechanisms \cite{Nisan:2001aa,lu2008randomized,mu2007setting},
fractional scheduling \cite{christodoulou2010mechanism}, Bayesian scheduling
\cite{Chawla2013a,DW15,gkyr2015-wine} or restricted domains where the processing
times come from discrete sets \cite{lavi2009truthful}. Alongside the approximation
ratio results, there has also been work on structural properties and
characterizations
\cite{dobzinski2008characterizations,christodoulou2008characterization}. For a more
detailed exposition of some of these results, we refer the reader to the survey
of~\citet{christodoulou2009mechanism}.

\subsubsection{The Truthful Setting vs the Strategic Setting}\label{sec:truthful-vs-strategic}

As we mentioned earlier, given that truthfulness is a very demanding requirement
which imposes strict constraints on the allocation and payment functions, it is an
interesting direction to consider whether \emph{non-truthful} mechanisms could
perform better, when evaluated in the worst-case equilibrium. In other words, for a
given problem, one could ask the following question:
\begin{quote}\emph{``Do there exist (non-truthful) mechanisms whose Price of Anarchy
outperforms the approximation ratio guarantee of all truthful mechanisms?''}.
\end{quote} To differentiate, we will refer to the traditional approach
of~\citet{Nisan:2001aa} as the \emph{truthful setting} and to the setting where all
mechanisms are explored (with respect to their Nash equilibria) as the
\emph{strategic setting}.

\Citet{K14} studied the truthful setting for the problem of unrelated machine
scheduling \emph{without money} but he explicitly advocated the strategic setting as
a future direction. This was later pursued in \citet{GKK16} for the same problem,
where the authors answered the aforementioned question in the affirmative. The same
approach was taken in \cite{christodoulou2016social} following the results of
\cite{filos2014truthful} on the limitations of truthful mechanisms for indivisible
item allocation. In the literature of auctions, the strategic setting was studied
even in domains for which an optimal truthful mechanism (the VCG mechanism) exists,
motivated by the fact that non-truthful mechanisms are being employed in practice,
with the Generalized Second-Price auction used by Google for the Adwords allocation
being a prominent example \cite{caragiannis2015bounding}. We refer the reader to the
survey of~\citet{roughgarden2017price} for more details.

Somewhat surprisingly, although the exploration of different solution concepts
besides dominant strategy equilibria was already explicitly mentioned as a future
direction by~\citet{Nisan:2001aa}, the strategic setting for
the scheduling problem was not studied before our paper. As we mentioned earlier,
the answer to the highlighted question above here is negative, but the setting
proved out to be quite rich in terms of the achievable trade-off between the two
different inefficiency notions.

To the best of our knowledge, ours is the first paper that proposes the systematic
study of the trade-off between the Price of Anarchy and the Price of Stability.
While preparing our manuscript, we became aware that a trade-off between the two
notions was very recently considered also in \citet{ramaswamy2017impact}, though in
a fundamentally different setting: the authors of \cite{ramaswamy2017impact} study a
special case of covering games, originally introduced
by~\citet{gairing2009covering}, which is not inherently a mechanism design setup.
On the contrary, our interest is in explicitly studying this trade-off in the area
of algorithmic mechanism design, thus choosing the prototypical scheduling problem
as the starting point.

\section{Model and Notation}
\label{sec:prelims}

Let $\Rgeq=[0,\infty)$ denote the nonnegative reals and $\N=\sset{1,2,\dots}$ the
positive integers. For any $n\in\N$, let $[n]=\ssets{1,2,\dots,n}$. In the
\emph{strategic scheduling} problem (on unrelated machines), there is a set
$N=\{1,\ldots,n\}$ of \emph{machines} (or agents) and a set $J=\{1,\ldots,m\}$ of
\emph{tasks}. Each machine $i$ has a \emph{processing time} (or \emph{cost})
$t_{i,j}\geq 0$ for task $j$. The induced matrix $\t\in\Rgeq^{n\times m}$ is the
\emph{profile} of processing times. For convenience, we will denote by $\t_i =
(t_{i,1}, \ldots, t_{i,m})$ the vector of processing times of machine $i$ for the
tasks and by $\t^j = (t_{1,j},\ldots, t_{n,j})$ the vector of processing
times of the machines for task $j$, so that $\t = (\t_1, \ldots, \t_n) =
(\t^1,\ldots,
\t^m)^\top$.
The machines are
\emph{strategic} and therefore, when asked, they do not necessarily report their
true processing times $\t$ but they rather use \emph{strategies} $\ss\in\Rgeq^{n\times m}$.
To emphasize the distinction, we will often refer to $\t$ as the profile of
\emph{true} processing times. Adopting standard game-theoretic notation, we use
$\vecc t_{-i}$ and $\vecc s_{-i}$ to denote the profile of true or reported
processing times respectively, without the coordinates of the $i$'th machine.

A (deterministic, direct revelation) \emph{mechanism} $\mech=(\vecc x,\vecc p)$ gets
as input a strategy profile $\ss\in\R^{n\times m}$ reported by the machines and
outputs \emph{allocation} $\vecc x=\vecc x(\ss)\in\sset{0,1}^{n\times m}$ and
\emph{payment} $\vecc p=\vecc p(\ss)\in\Rgeq^{n}$: $x_{i,j}$ is an indicator
variable denoting whether or not task $j$ is allocated to machine $j$, and $p_i$ is
the payment with which $\mech$ compensates machine $i$ for taking part in the
mechanism. Thus, the allocation rule needs to satisfy $\sum_{i\in N}x_{i,j}(\vecc s)
= 1$ for all tasks $j$.

The \emph{utility} of machine $i$ under a mechanism $\mech=(\vecc x, \vecc p)$,
given true running times $\t_i$ and a reported profile $\ss$ by the machines, is
$$
u_i^{\mech}(\ss|\t_i)= p_i(\ss) - \sum_{j=1}^mx_{i,j}(\vecc s)t_{i,j},
$$
that is, the payment she receives from $\mech$ minus the total workload she has to
execute. This is exactly the reason why machines may lie about their true processing
times; they will change their report $\ss_i$ and deviate to another $\ss_{i}'$ if this
improves the above quantity. A stable solution with respect to such best-response
selfish behaviour is captured by the well-known notion of an equilibrium.
Given a mechanism $\mech$ and a strategy profile $\ss$, we will say that $\ss$ is a
\emph{(pure Nash) equilibrium}\footnote{We will be interested in pure Nash
equilibria in this paper, but we discuss different solution concepts
in~\cref{sec:solutionconcepts} as well as in~\cref{sec:discussion}.} of
$\mech$ (with respect to a true profile $\t$) if, for every machine $i$ and every
possible deviation $\ss_i'\in\Rgeq^m$,
$$
u_i^{\mech}(\ss|\t) \geq u_i^{\mech}(\ss_i',\ss_{-i}|\t).
$$
Let $\mathcal{Q}_\t^\mech$ denote the set of pure Nash equilibria of mechanism
$\mech$ with respect to true profile $\vecc t$. Following the related literature (see, 
e.g., \cite{Nisan:2001aa,christodoulou2016social,GKK16}), we
will consider mechanisms for which Nash equilibria exist for every profile of processing times, i.e., $\mathcal{Q}_\t^\mech\neq\emptyset$ for all $\vecc
t\in\Rgeq^{n\times m}$.

Our objective is to design mechanisms that minimize the \emph{makespan}
$$
C^\mech(\ss|\t)=\max_{i\in N} \sum_{j=1}^mx_{i,j}(\vecc s)t_{i,j},
$$
that is, the total completion time if our machines run in parallel. For a matrix
$\t$ of running times, let $\opt(\t)$ denote the optimum makespan, i.e., $\opt(\vecc
t)=\min_\vecc y \max_{i \in N} \sum_{j=1}^my_{i,j}t_{i,j}$ where $\vecc y$ ranges over all feasible
allocation of tasks to machines. It is a well-known phenomenon that equilibria can
result in suboptimal solutions, and the following, extensively studied, notions
where introduced to quantify exactly this discrepancy: the \emph{Price of Anarchy}
(PoA) and the \emph{Price of Stability} (PoS) of a scheduling mechanism $\mech$ on
$n$ machines are, respectively,
$$
\poa(\mech)=\sup_{m\in\N,\t\in\Rgeq^{n\times m}}\frac{\sup_{\vecc s\in\mathcal{Q}^\mech_\t} C^\mech(\ss|\t)}{\opt(\vecc t)}
\qquad
\pos(\mech)=\sup_{m\in\N,\t\in\Rgeq^{n\times m}}\frac{\inf_{\vecc s\in\mathcal{Q}^\mech_\t} C^\mech(\ss|\t)}{\opt(\vecc t)}.
$$

For simplicity, we will sometimes drop the $\mech$, $\t$ and $\ss$ in the notation
introduced in this section, whenever it is clear which mechanism and which true or
reported profile we are referring to.

\subsection{Task-Independent Mechanisms}

For a significant part of this paper, we will focus on the class of anonymous,
task-independent mechanisms. This is a rather natural class of mechanisms; as a
matter of fact, two of the arguably most well-studied and used mechanisms in
practice, namely the First-Price and Second-Price, lie within this class.

\begin{definition}[Task-independence]\label{def:task-ind}
  A mechanism $\mech=(\vecc x,\vecc p)$ is called \emph{task-independent} if
  each one of its tasks is allocated independently of the others.
  Formally, there exists a collection of single-task mechanisms
  $\{\AA_j\}_{j=1,\dots,m}$, $\AA_j=(\vecc y^j,\vecc q^j)$, such that,
  for any task $j$, any machine $i$, and for any strategy profile $\ss$,
	$$
	\vecc x^j(\vecc s)=\vecc y^j(\vecc s^j) \qquad \text{and}
        \qquad p_i(\vecc s) =\sum_{j=1}^m q^j_i(\vecc s^j).
	$$
\end{definition}
We will refer to the single-task mechanisms $\AA_j$ of the above
definition as the \emph{components} of $\mech$. It is important to
notice here that the definition does not require the mechanism to
necessarily use the same component for all the tasks.

Another standard property in the literature of the problem is anonymity. The
property can be defined generally (e.g., see \cite{K14,ashlagi2012optimal}), but
here we will define it for task-independent mechanisms. Since we are dealing with
potentially non-truthful mechanisms, our notion of anonymity needs to refer to the
equilibria of the mechanism.
\begin{definition}[Anonymity]\label{def:anonymity}
   A single-task mechanism $\AA=(\vecc x,\vecc p)$ is \emph{anonymous}
  if, for any true processing time profile $\t$ with no ties\footnote{That is, $t_i \neq t_{i'}$ for all $1 \le i\neq {i'} \le n$.}
  and any permutation\footnote{For any permutation
  $\pi:\ssets{1,\dots,n}\map\ssets{1,\dots,n}$ and $n$-dimensional vector $\vecc
  x=(x_1,\dots,x_n)$, the permutation of $\vecc x$ under $\pi$ is the vector
  $\pi(\vecc x)\equiv (x_{\pi(1)},\dots,x_{\pi(n)})$.} $\pi$, if there exists an
  equilibrium $\ss$ under $\t$, then there exists an equilibrium $\tilde \ss$ under
  true profile $\pi(\t)$ with allocation $\vecc{x}(\tilde \ss) =
  \pi(\vecc{x}(\ss))$.
  A task-independent mechanism $\mech$ is anonymous, if all its
   components are anonymous (single-task) mechanisms.
\end{definition}

\begin{remark}
\label{rem:anonymity}
Our notion of anonymity refers to the true profiles, and stipulates that after any permutation
of machine identities, a corresponding equilibrium exists. This is
the natura analogue of anonymity for the case of Nash equilibria; indeed, if one substitutes the notion of ``Nash equilibrium'' by ``dominant strategy equilibrium'' in \cref{def:anonymity}, then the standard notion employed by \citet{ashlagi2012optimal} for truthful mechanisms is recovered. Note that similarly to \cite{ashlagi2012optimal}, we only require this property to hold when the profiles of processing times do not exhibit ties.
\end{remark}

Perhaps the simplest and most natural mechanism that one can think of
is the following, which assigns the task to the fastest machine
(according to the declared processing times) and pays her her
declaration.
\begin{definition}[First-Price (FP) mechanism]
  \label{def:FPmech}
  Assign each task $j$ to the fastest machine $\iota(j)$ for it,
  i.e. $\iota(j) \in \arg\min_{i \in N} s_{i,j}$ (breaking ties
  arbitrarily), paying her her declared running time $s_{\iota(i),j}$;
  pay the remaining $N\setminus \ssets{\iota(j)}$ machines $0$ for task
  $j$.
\end{definition}

Second-Price mechanisms have also been extensively studied and applied in auction
theory, but also in strategic scheduling. As we mentioned in the introduction, the
following mechanism is usually referred to as the VCG mechanism in the literature of
the problem (see e.g., \cite{christodoulou2009mechanism}):
\begin{definition}[Second-Price (SP) mechanism]
\label{def:SPmech}
	Assign each task $j$ to the fastest machine $\iota(j)$ for it, i.e., $\iota(j)
	\in \arg\min_{i \in N} s_{i,j}$ (breaking ties arbitrarily), paying her the
	declared processing time of the second-fastest machine, i.e. $\min_{i \in N
	\setminus \ssets{\iota(j)}} s_{i,j}$; pay the remaining $N\setminus
	\ssets{\iota(j)}$ machines $0$ for task $j$.
\end{definition}

Notice that both FP and SP mechanisms are task-independent and anonymous.
Furthermore, SP is truthful. As a matter of fact, SP is the best known truthful
mechanism if one is interested only in dominant strategy equilibria~(see,
e.g.,~\cite{Nisan:2001aa,christodoulou2009lower} and \cref{sec:solutionconcepts}).

\subsection{Solution Concepts and Notions of Inefficiency}
\label{sec:solutionconcepts}

The solution concept that we consider in this paper is that of the pure Nash
equilibrium. In the literature of the truthful scheduling problem, the employed
solution concept is that of the \emph{dominant strategy} equilibrium, i.e., a
strategy profile in which no agent would have an incentive to deviate to any
other strategy, no matter the strategies of the remaining agents. More
precisely, the literature has been interested in \emph{truthful} mechanisms, i.e.,
mechanisms for which truth-telling is always (i.e., for any processing time profile
$\t$) a dominant strategy equilibrium. The goal is to find a mechanism with the best
\emph{approximation ratio}, which is defined as the worst-case (over all inputs)
ratio of the makespan of the mechanism over the optimal makespan, in the
truth-telling equilibrium.

For this objective, studying only the truth-telling dominant strategy equilibria is
without loss of generality, by the Revelation Principle (see, e.g., \cite{Nisan07}). We
remark however that, a priori, the fact that an allocation function can be implemented in
truth-telling dominant strategies (i.e., an appropriate payment function can be
found such that the resulting mechanism is truthful) does not have any
implications on the space of non-truthful mechanisms and their PoA/PoS guarantees.

There are however some inherent relations between the approximation ratio, the Price
of Anarchy and the Price of Stability which follow directly from their definitions.
Clearly, a Price of Anarchy guarantee is stronger than a Price of Stability
guarantee, since the former bounds the inefficiency of all equilibria while the
latter is only concerned with the best one. Since dominant strategy equilibria are
also Nash equilibria by definition, for truthful mechanisms, a Price of Anarchy
guarantee is also stronger than an approximation ratio guarantee, which, in turn, is
stronger than a Price of Stability guarantee. An illustration of the relation
between these different notions is given in \cref{fig:solutionconcepts}.

\begin{figure} 
	\centering    
	\begin{tikzpicture}[font=\footnotesize]
 \draw [draw=black, fill=gray, opacity=0.2] (0,0) ellipse (3cm and 2cm);

 \draw [draw=black, fill=blue, opacity=0.8] (0,0.3) circle (0.1cm and 0.1cm);
 \draw [draw=black, fill=blue, opacity=0.8] (0.2,1.5) circle (0.1cm and 0.1cm);
 \draw [draw=black, fill=blue, opacity=0.8] (-0.5,0.9) circle (0.1cm and 0.1cm);
 \draw [draw=black, fill=blue, opacity=0.8] (0.7,0.6) circle (0.1cm and 0.1cm);
 \draw [draw=black, fill=blue, opacity=0.8] (-0.8,0.15) circle (0.1cm and 0.1cm);
 \draw [draw=black, fill=blue, opacity=0.8] (-1.23,0.38) circle (0.1cm and 0.1cm);

 \draw [draw=black, fill=blue, opacity=0.8] (0.5,-0.1) circle (0.1cm and 0.1cm);
 \draw [draw=black, fill=blue, opacity=0.8] (0.7,-0.6) circle (0.1cm and 0.1cm);
 \draw [draw=black, fill=blue, opacity=0.8] (1.4,-0.2) circle (0.1cm and 0.1cm);
 \draw [draw=black, fill=blue, opacity=0.8] (-0.8,-0.4) circle (0.1cm and 0.1cm);
 \draw [draw=black, fill=blue, opacity=0.8] (1,-1.1) circle (0.1cm and 0.1cm);

 \draw[->] (-1.23,0.38) -- (-3.23,1);
 \node[above] at (-3.23,1) {Best equilibrium (PoS)};

\draw[->] (0,0.3) -- (3.0,1);
\node[above] at (3.5,1) {\begin{minipage}{3.5cm}
Truth-telling profile $\t$\\(Approximation Ratio)
\end{minipage}};

\draw[->] (1.4,-0.2) -- (2.9,-1.5);
\node[below] at (3.3,-1.4) {Worst equilibrium (PoA)};

\draw [dashed] plot [smooth cycle] coordinates {
(-1.4,-0.)
(-1.,-0.7)
(-0.3,-0.6)
(1.25,-1.4)
(1.73,-0.7)
(0.3,1.8)
(-1.05,1)
};

 \draw[->] (0.3,-0.9) -- (-0.9,-2.3);
\node[below] at (-0.9,-2.3) {All equilibria};

\draw[->] (-3.2,-3) -- (5,-3);
\draw[dashed] (-3,-3.0) -- (-3,0);
\node[below] at (-3.0,-3) {$1$};
\node[below] at (5.0,-3.0) {$\infty$};
\node[above] at (5.0,-3.0) {Inefficiency};
 
\end{tikzpicture}
	\caption{A pictorial representation of the relation between the different
	solution concepts and the notions of inefficiency.
	The blue nodes represent the set of equilibria of a mechanism $\mech$ for a
	fixed true underlying profile of processing times $\t$ (they are depicted as a
	finite set, for convenience, but this need not be the case).
	The wider grey area is the set of all feasible input strategy profiles $\ss$ of
	$\mech$.
	We have marked the best and worst (under $\t$) equilibria (assuming they are
	unique, for ease of exposition), as well as the truth-telling profile. Note that
	if $\mech$ is truthful, then the truth-telling profile is a (dominant strategy)
	equilibrium (but, in general, this profile might not even belong to the set of
	equilibria).
	The $\poa$ bounds the inefficiency of all the blue nodes, the $\pos$ bounds the
	inefficiency of the left-most node, and the approximation ratio bounds the
	inefficiency of the truth-telling node; all these bounds are computed, in the
	worst case, over all possible true profiles $\t$).}
	\label{fig:solutionconcepts}
\end{figure}
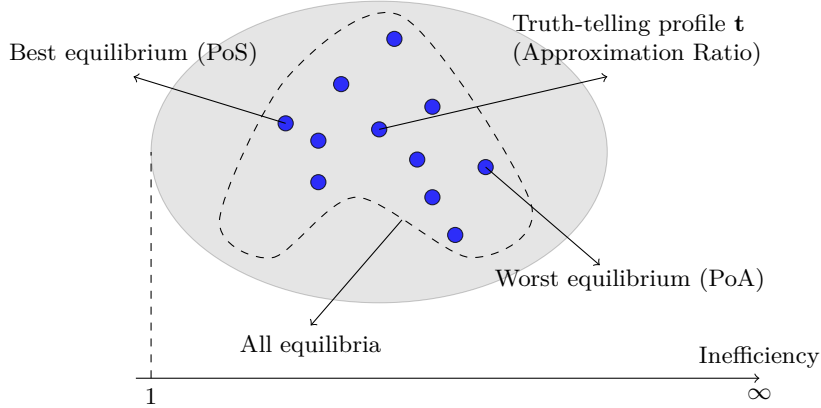

\section{The Inefficiency of \emph{All} Mechanisms}
\label{sec:pareto_all}

We start with a lower bound of $n$ for the Price of Anarchy of the scheduling
problem, which applies to \emph{all} mechanisms. The lower bound will be based on
the following monotonicity lemma. We note that this monotonicity property is
different from the weak monotonicity (WMON) used in the literature of truthful
machine scheduling (see e.g.,~\citep{Nisan:2001aa,christodoulou2009lower}), in the
sense that (a) it is global, whereas WMON is local and (b) it applies to the
relation between the true processing times and the equilibria of the mechanism,
rather than the actual allocations.

\begin{lemma}[Equilibrium Monotonicity]
\label{lem:monotonicity}
	Let ${\mech}$ be any mechanism for the scheduling problem. Let
	$\t$ be a profile of true processing times and let
	$\ss \in \mathcal{Q}_\t$ be an equilibrium under $\t$.
	Denote by $S_i$ the set of tasks assigned to machine $i$ by $\mech$ on input $\ss$.
	Consider any profile $\hat{\t}$
	such that for every machine $i$, $\hat{t}_{i,j} \leq t_{i,j}$ if $j \in S_i$ and
	$\hat{t}_{i,j} \geq t_{i,j}$ if $j \notin S_i$.
	Then $\ss\in\mathcal{Q}_{\hat{\t}}$, i.e., $\ss$ is an
	equilibrium under $\hat{\t}$ as well.
\end{lemma}

\begin{proof}
	Assume by contradiction that $\ss \notin \mathcal{Q}_{\hat{\t}}$, which means
	that for the profile of processing times $\hat{\t}$, there exists some machine
	$i$ that has a beneficial deviation $\vecc s_i'$, i.e.,
	$u_i(\ss_i',\ss_{-i}|\hat{\t}) > u_i(\ss|\hat{\t})$.
	Let $S_i'$ be the set of tasks assigned to machine
	$i$ under report $\ss'=(\ss_i',\ss_{-i})$ (and underlying true reports $\hat{\t}$).
	The difference in utility for machine $i$
	between profiles $\ss'$ and $\ss$ is
	\[ \Delta u_i(\hat{\t})
	\equiv u_i(\ss'|\hat{\t}) - u_i(\ss|\hat{\t})
	= p_i(\ss') - p_i(\ss) + \sum_{j\in S_i \backslash S_i'} \hat{t}_{i,j} - \sum_{j \in S_i'\backslash S_i} \hat{t}_{i,j}.
	\]
	By the fact that $s_i'$ is a beneficial deviation, it holds that $\Delta u_i(\hat{\t}) >0$.

	Now consider the profile of processing times $\t$ and the same deviation $\vecc s_i'$ of machine $i$. The increase in utility is
	\[ \Delta u_i(\t)= p_i(\ss') - p_i(\ss) + \sum_{j\in S_i \backslash S_i'} t_{i,j} - \sum_{j \in S_i'\backslash S_i} t_{i,j}
		\geq p_i(\ss') - p_i(\ss) + \sum_{j\in S_i \backslash S_i'} \hat{t}_{i,j} - \sum_{j \in S_i'\backslash S_i} \hat{t}_{i,j}
		= \Delta u_i(\hat{\t}),
	\]
	which holds because $t_{i,j} \geq \hat{t}_{i,j}$, if $j \in S_i$ and $t_{i,j} \leq \hat{t}_{i,j}$, if $j \notin S_i$.
	This implies that $\Delta u_i(\t) >0$, which contradicts the fact that $\ss\in\mathcal{Q}_{\t}$.
\end{proof}
Using this lemma, we can prove our first lower bound:
\begin{theorem}
\label{thm:lower}
	For any scheduling mechanism $\mech$ for $n$ machines, it must be that $\poa(\mech) \geq n$.
\end{theorem}

\begin{proof}
Let ${\mech}$ be any mechanism and consider a profile of true processing times
$\t$ with $n$ machines and $n^2$ tasks, where $t_{i,j}=1$ for all machines $i$
and all tasks $j$. Let $\ss=(s_1,s_2,\ldots,s_n)$ be a pure Nash equilibrium of
${\mech}$ under $\t$. For each machine $i$, let $S_i$ be the set of tasks
assigned to that machine and note that there exists some machine $k$ for which
$|S_k|\geq n$. Let $T_{k} \subseteq S_{k}$ be any subset of $S_k$ such that
$|T_k| =n$.

Now consider the following profile $\hat{t}$ of processing times:
 \begin{itemize}
 	\item For all $i \neq k$, $\hat{t}_{i,j} = 0$, for all $j \in S_i$ and $\hat{t}_{i,j}=t_{i,j}$, for all $j \notin S_i$.
 	\item $\hat{t}_{kj}=0$, for all $j \in S_k \backslash T_{k}$ and $\hat{t}_{kj}=t_{k,j}$, for all $j \notin S_k \backslash T_k$.
 \end{itemize}
By \cref{lem:monotonicity}, the profile $\ss=(s_1,s_2,\ldots,s_n)$ is a pure
Nash equilibrium under $\hat{\t}$ and the allocation is the same as before, for a
makespan of at least $n$, since machine $k$ is assigned all the tasks in $T_k$. The
optimal allocation will assign one task from $T_k$ to each machine, the tasks from
$S_i$ to machine $i$ for each $i \neq k$ and the tasks from $S_k \backslash T_k$ to
machine $k$, for a total makespan of $1$ and the Price of Anarchy bound follows.
\end{proof}

\subsection{PoA/PoS Trade-off}
\label{sec:main_theorem_all}

In this section, we prove our main theorem regarding the trade-off between the Price
of Anarchy and the Price of Stability. The theorem informally says that if the Price
of Anarchy of a mechanism is small, then its Price of Stability has to be high.

\begin{theorem}\label{thm:tradeoff}
	\label{th:PoS_PoA_tradeoff}
	For any scheduling mechanism $\mech$ for $n$ machines, and any positive real $\alpha$,
	$$
	\poa (\mech)< \alpha
	\quad \then \quad
	\pos (\mech) \geq \frac{n-1}{\alpha}+1.
	$$
\end{theorem}

\begin{proof}
By performing the transformation $\rho=\frac{n-1}{\alpha}+1$ and taking the
contrapositive, it is not difficult to see that we need to prove that
$$
	\pos (\mech) < \rho \quad \then \quad  \poa (\mech)\geq \frac{n-1}{\rho -1},
	$$
for any real $\rho >1$.

Consider an instance with $n$ agents and $n$ tasks. Assume a true $n\times n$
processing-times matrix $\vecc t$ with
$$
t_{1,j} =
\begin{cases}
n-1, &\text{if}\;\; j=1,\\
\rho -1, &\text{otherwise},
\end{cases}
$$
and
$$
t_{i,j} =
\begin{cases}
n-1, &\text{if}\;\; j=i,\\
\infty, &\text{otherwise},
\end{cases}
$$
for all $i=2,\dots,n$. Here $\infty$ denotes an arbitrarily large positive value,
and actually replacing it with any value $M \geq\rho (n-1)$ will work just fine for
our proof.\footnote{We will adopt a similar convention throughout the paper.}

First notice that by allocating each task $j$ to machine $j$ with running time
$t_{j,j}=n-1$, for all $j\in [n]$, we get an upper bound of $n-1$ on the optimal
makespan of $\vecc t$. Thus, since $\pos(\mech)< \rho$, there must exist a pure Nash
equilibrium profile $\vecc{s}^\star$ such that the allocation
$\mech(\vecc{s}^\star)$ results in a makespan less than $\rho (n-1)$ (with respect
to the underlying, true time matrix $\vecc t$). But then, due to the structure of
$\vecc t$, and in particular the large value of $M$, $\mech(\vecc{s}^\star)$ can
only allocate each task $j$ to either machine $1$ or machine $j$, for all $j\in [n]$.
In particular, task $1$ will necessarily have to be allocated to machine $1$.
Furthermore, from the remaining $n-1$ tasks, not all of them can be allocated to
machine $1$, because that would give rise to a running time of $
t_{1,1}+\sum_{j=2}^n t_{1,j} =n-1+(n-1)(\rho-1)=\rho (n-1)$ for machine $1$, which
violates the Price of Stability constraint assumed for $\vecc{s}^\star$. So, there
must exist at least one task $j\geq 2$, denote it by $j^\star$, such that
$\mech(\vecc{s}^\star)$ allocates $j$ to machine $j$.

For each task $j$, let $i_j$ denote the machine which task $j$ is allocated to by
$\mech(\vecc s^\star)$.  Now modify the original, true execution time matrix $\vecc
t$ by changing the running time $t_{i_j,j}$, for all $j\neq j^\star$, to
$t'_{i_j,j}=0$. Denote this new matrix by $\vecc t'$. Due to
\cref{lem:monotonicity}, $\vecc{s}^\star$ has to be a pure Nash equilibrium of
$\mech$ with respect to the modified true profile $\vecc t'$ as well. But now
$\mech(\vecc{s}^\star)$ results in a makespan of at least
$t'_{j^\star,j^\star}=t_{j^\star,j^\star}=n-1$ (since task $j^\star$ is allocated to
machine $j^\star$), while allocating $j^\star$ to machine $1$ (and leaving all other
assignments as they are, i.e.\ task $j\neq j^\star$ gets allocated to machine $i_j$)
results in machine $1$ having a total running cost of at most $(n-1)\cdot 0 +
t'_{1,j^\star}= \rho -1$, and all other machines
$0$. This gives a Price of Anarchy lower bound of
$\frac{n-1}{\rho -1}$.
\end{proof}

By allowing $\alpha$ in \cref{th:PoS_PoA_tradeoff} to grow arbitrarily large, we get the following:
\begin{corollary}\label{cor:pos1}
	\label{th:PoSopt_PoAinfinity}
	Even for just two machines, if a scheduling mechanism has an optimal Price of
	Stability of $1$, then its Price of Anarchy has to be unboundedly large.
\end{corollary}

From the results of this section, as well as the trivial fact that $\poa({\mech})
\geq \pos(\mech)$ for any mechanism $\mech$, we obtain a feasibility trade-off
between the PoA and the PoS of scheduling mechanisms, which is illustrated in
\cref{fig:pareto_general}. We refer to the boundary of the shaded feasible region as
the \emph{inefficiency boundary}; the shape of the boundary follows from
\cref{thm:tradeoff}, as well as \cref{thm:lower}, since for
$\pos(\mech)>2-\frac{1}{n}$ (or, in the language of~\cref{thm:tradeoff}, for
$\alpha<n$), the best (i.e.~largest) lower bound on the PoA is now given
by~\cref{thm:lower}.

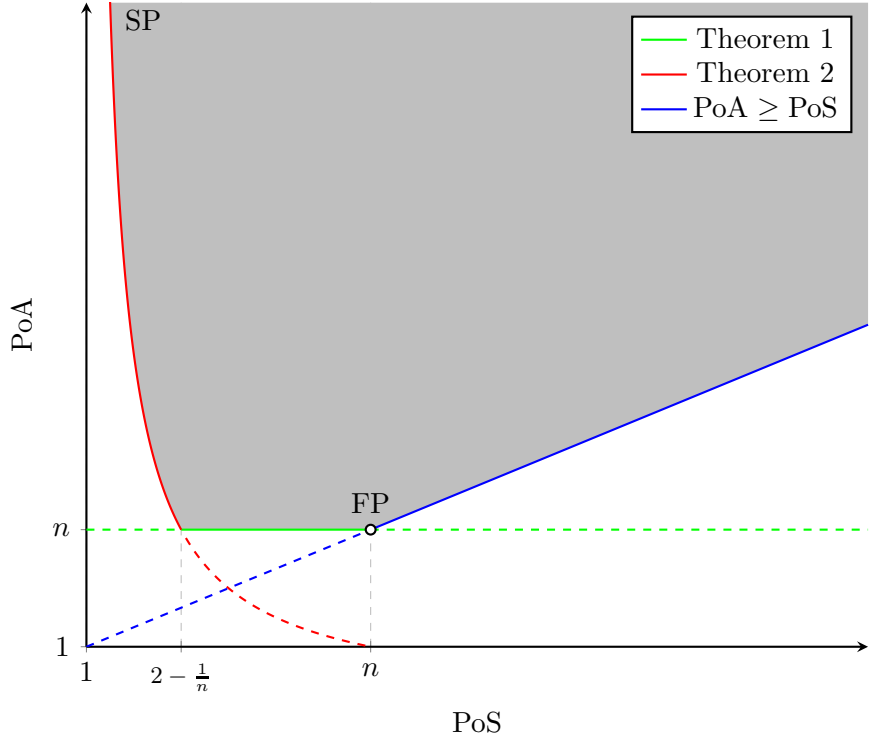
\begin{figure} 
\centering    
\begin{tikzpicture}
\begin{axis}[
 	scale=1.5,
	xmin=1, xmax=6.5,
	ymin=1, ymax=12,
    axis lines = left,
    ytick={1,3},
    yticklabels={$1$,$n$},
    xtick={1,1.6667,3},
    xticklabels={$1$,{\footnotesize$2-\frac{1}{n}$},$n$},
    xmajorgrids=true,
    grid style=dashed,
    xlabel = {$\pos$},
    ylabel = {$\poa$},
    thick
]
\path[fill=lightgray,smooth]
(1.16667,12.) -- (1.17667,11.3208) -- (1.18667,10.7143) -- (1.19667,10.1695) --
(1.20667,9.67742) -- (1.21667,9.23077) -- (1.22667,8.82353) -- (1.23667,8.4507) --
(1.24667,8.10811) -- (1.25667,7.79221) -- (1.26667,7.5) -- (1.27667,7.22892) --
(1.28667,6.97674) -- (1.29667,6.74157) -- (1.30667,6.52174) -- (1.31667,6.31579) --
(1.32667,6.12245) -- (1.33667,5.94059) -- (1.34667,5.76923) -- (1.35667,5.60748) --
(1.36667,5.45455) -- (1.37667,5.30973) -- (1.38667,5.17241) -- (1.39667,5.04202) --
(1.40667,4.91803) -- (1.41667,4.8) -- (1.42667,4.6875) -- (1.43667,4.58015) --
(1.44667,4.47761) -- (1.45667,4.37956) -- (1.46667,4.28571) -- (1.47667,4.1958) --
(1.48667,4.10959) -- (1.49667,4.02685) -- (1.50667,3.94737) -- (1.51667,3.87097) --
(1.52667,3.79747) -- (1.53667,3.72671) -- (1.54667,3.65854) -- (1.55667,3.59281) --
(1.56667,3.52941) -- (1.57667,3.46821) -- (1.58667,3.40909) -- (1.59667,3.35196) --
(1.60667,3.2967) -- (1.61667,3.24324) -- (1.62667,3.19149) -- (1.63667,3.14136) --
(1.64667,3.09278) -- (1.65667,3.04569) -- (1.66667,3.) -- (3,3) --
(\pgfkeysvalueof{/pgfplots/xmax},\pgfkeysvalueof{/pgfplots/xmax}) --
(\pgfkeysvalueof{/pgfplots/xmax},\pgfkeysvalueof{/pgfplots/ymax}) --
(1.667,\pgfkeysvalueof{/pgfplots/ymax}) -- cycle
;
\addplot [green,domain=1:1.667,dashed,samples=3]
{3};
\addplot [green,domain=1.667:3,samples=3]
{3};
\addplot [green,domain=3:6.5,dashed,samples=3]
{3};
\addplot [red,domain=1.15:1.667,samples=200]
{(3-1)/(x-1)};
\addplot [red,domain=1.7:3,dashed,samples=200]
{(3-1)/(x-1)};
\addplot [blue,domain=1:3,dashed,samples=3]
{x};
\addplot [blue,domain=3:6.5,samples=3]
{x};
\node[circle,draw,fill=white,label=above:{FP},inner sep=1.3pt] (fp) at (3,3) {};
\node (sp) at (1.4,11.7) {SP};
\legend{,\cref*{thm:lower},,\cref*{thm:tradeoff},,,$\poa\geq\pos$}
\end{axis}
\end{tikzpicture}
\caption{The inefficiency boundary for general mechanisms, given by
\cref{thm:tradeoff} (red line). Combined with the global PoA lower bound of
\cref{thm:lower} (green line) and the trivial fact that the PoS is at most the PoA
(blue line), we finally get the grey feasible region.}
\label{fig:pareto_general}
\end{figure}

\subsubsection{Mechanisms on the Extrema of the Inefficiency Boundary}
\label{sec:extema}
When looking for mechanisms on the Pareto frontier, the first ones that come to mind
are perhaps the First-Price (FP) and Second-Price (SP) mechanisms, defined in \cref{sec:prelims},
which are straightforward adaptations of the well-known First-Price auction and Second-Price auction
mechanisms from the auction literature.

It follows from known results in the literature for the First-Price auction (see,
e.g., \cite{feldman2016correlated}) that in every pure Nash equilibrium of the FP,
each task is allocated to the machine with the smallest \emph{true} processing time
for the task; we provide a simple proof below for completeness. In \cref{sec:spa_mechs}, 
we will define a class of task-independent
mechanisms ($\mechsp_\alpha$) that contain FP as a corner case ($\mechsp_1$).

\begin{theorem}\label{thm:FP}
The $\poa$ and the $\pos$ of the First-Price mechanism are both $n$.
\end{theorem}

\begin{proof}
First, we argue that in every equilibrium of $\mathrm{FP}$, for any task $j$ a
machine with the fastest true processing time for $j$ receives the task. Given any
profile of true processing times $\mathbf{t}$, let $J_a$ be the set of machines with
the fastest true processing time $t_f$ for task $j$. Assume by contradiction that
some machine $k \notin J_a$ receives the task at some equilibrium $\ss^j$. Since
$\ss^j$ is an equilibrium, it must be the case that $p_{k,j}=s_{k,j} \geq t_{k,j} > t_f$, as otherwise
machine $k$ would have negative utility. But then, by the continuity of the
strategy space, any machine $i \in J_a$ can report $s_{i,j}' \in (t_f,s_{k,j})$ and
win the task, obtaining positive utility. This contradicts the fact that $\ss^j$
is an equilibrium.

Given this, it is not hard to see that $\poa(\mathrm{FP}) = n$, as in the
worst-case, every task will go to the same machine, which will be the fastest
machine for all tasks. To show that $\pos(\mathrm{FP}) = n$, it is easy to construct
an example where there is a single machine $k$ that has the fastest processing time
$t_f$ for each task, and $n-1$ different machines $i_1, \ldots, i_{n-1}$, such that
machine $i_j$ has processing time $t_f+\varepsilon$ for task $j$, where $\varepsilon$ can
be arbitrarily small. In every equilibrium, all tasks go to machine $k$ for a
makespan of $n\cdot t_f$, whereas in the optimal schedule, machine $k$ receives task $1$ and
machine $i_j$ receives task $j$. The Price of Stability goes to $n$ as $\varepsilon\to
0$.
\end{proof}

For the Second-Price mechanism, again it follows from known observations in the literature
that while the mechanism is truthful, it has several other
pure Nash equilibria as well. More precisely, for a task $j$ and any machine $i$, there exists
an equilibrium for which task $j$ is allocated to machine $i$. Therefore, we have the following:
\begin{theorem}\label{thm:vcg}
The $\poa$ of the Second-Price mechanism is unbounded and its $\pos$ is $1$.
\end{theorem}

\cref{thm:vcg} can be obtained as a corollary of our results in
\cref{sec:spa_mechs}, since SP is also a corner-case mechanism in our class, namely
$\mechsp_\infty$. Interestingly, as we identify in~\cref{thm:PoAtruthful} below,  it
turns out that the bad PoA bound is a inherent characteristic of all truthful
mechanisms. In other words, if one is interested in the set of \emph{all}
equilibria, they would have to reach out beyond truthful mechanisms.

\begin{theorem}
	\label{thm:PoAtruthful}
	The Price of Anarchy of any truthful mechanism is unbounded. 
\end{theorem}	
\begin{proof}
Let $\mech=(\vecc x,\vecc p)$ be a truthful mechanism on $n\geq 2$ machines. To
	arrive to a contradiction, assume that there exists a real $M\geq 1$ such that
	$\poa(\mech)\leq M$. We will consider single-task instances with only $n=2$
	machines. This is without loss of generality, since one can add arbitrarily many more machines
	with $0$ running time for the task, and the proof remains valid. In particular, we
	assume an underlying vector of true running times $\t=(0,\varepsilon)$, where
	$\varepsilon\in(0,M^{-1})$, and a vector $\ss=(1,0)$ of reported costs. We first
	show that $\mech$ allocates the task to machine $i=2$, even if she reports slightly
	slower running times.

	\begin{claim}
		\label{claim:helper1}
		For any $\delta\in [0,\varepsilon]$, mechanism $\mech$ always allocates the
		task to machine $i=2$ on any input $\ss'=(1,\delta)$.
	\end{claim}

	\begin{proof}
		\renewcommand\qedsymbol{$\blacksquare$}
		First notice that, due to truthfulness, if we consider as true underlying profile
		$\t'=\ss'=(1,\delta)$, then $\ss'$ has to be an equilibrium. Next, for a
		contradiction, assume that there exists a nonnegative $\delta\leq\varepsilon$ such
		that $x_{2}(\ss')=0$. Then $x_1(\vecc s')=1$, and thus the makespan under
		equilibrium $\ss'$ (and true profile $\t'$) would be $1$, while an optimal solution
		would have given the task to the faster machine, for a makespan of $\delta$. This
		results in a PoA of at least $\frac{1}{\delta}\geq \varepsilon^{-1}>M$ which is a
		contradiction. This completes the proof of the claim.
	\end{proof}

	Using the same argument, we can also show that $\mech$ keeps allocating the task
	to machine $i=2$ as long as the other machine has a strictly positive cost:
	\begin{claim}
		\label{claim:helper3}
		For any $\delta>0$, mechanism $\mech$ always allocates the task to machine $i=2$ on any input $\ss'=(\delta,0)$.
	\end{claim}

	We are now ready to prove that $\ss$ is actually an equilibrium:
	\begin{claim}
		\label{claim:helper2}
		Reporting $\ss=(1,0)$ is an equilibrium of $\mech$ (under true costs $\t$).
	\end{claim}

	The above claim is enough to complete the proof, since by combining it with the
	previous \cref{claim:helper1} (with $\delta=0$) we get that the makespan of $\mech$
	under equilibrium $\ss$ is $\varepsilon>0$ (since the task goes to machine $i=2$)
	while the optimal one is $0$. This contradicts the fact that $\poa(\mech)$ is
	bounded.

	\begin{proof}[Proof of~\cref{claim:helper2}]
		\renewcommand\qedsymbol{$\blacksquare$}
		First we remark that, due to well-known characterizations of truthfulness for
		single-dimensional domains~\cite{Myerson1981a} (which apply to our case, since we
		have a single task), there exist real functions $h_1,h_2$ such that the utilities of
		our machines (with respect to true costs $\t$) on any vector of reports $\ss'$ are
		given by\footnote{See~\cite[Theorem~4.2]{archer2001truthful}.}
		\begin{equation}
		\label{eq:utilities_myerson}
		u_i(\ss') = h_i(s_{-i}') + (s_i'-t_i)\cdot x_i(\ss')-\int_{0}^{s_i'}x_i(z,s_{-i}')\,dz.
		\end{equation}
		Furthermore, the allocation function $x_i(\ss')$ of each agent $i$ is
		monotonically nonincreasing with respect to her reported cost $s_i'$.

		Now we show that the first agent has no incentive to deviate from reporting $s_1=1$
		as long as the second agent is fixed at $s_2=0$. Indeed, currently, and as long as
		she reports any strictly positive cost $s_1'=z>0$, she will still lose the task
		(i.e., $x_1(z,s_2)=0$). That holds due to \cref{claim:helper3}. From
		\eqref{eq:utilities_myerson}, this results in a utility of $u_1(s_1',s_2)=h_1(s_2)$.
		For the only remaining case that she reports $s_1=0'$, again
		from~\eqref{eq:utilities_myerson} we get that $u_1(0,s_2)=h_1(s_2)-(0-0)\cdot
		x_1(0,s_2)=h_1(s_2)$. Thus, in no case the first agent can gain by unilaterally
		deviating.

		Finally, we need to show that the second agent has no incentive to deviate from
		reporting $s_2=0$ as well. Assuming the other agent fixed at $s_1=1$, the
		improvement in her utility by declaring a cost $s_2'\geq 0$ is (due
		to~\eqref{eq:utilities_myerson})
		$$
		u_2(1,s_2')-u_2(1,0)=(s_2'-\varepsilon)x_2(1,s_2') - \int_{0}^{s_2'}x_2(1,z)\,dz.
		$$
		Recall now from \cref{claim:helper1} that $x_1(1,z)=1$ for all
		$z\in[0,\varepsilon]$. Thus, for $s_2'\in[0,\varepsilon]$ the above difference in
		the second agent's utility becomes
		$
		(s_2'-\varepsilon) - s_2'= -\varepsilon < 0,
		$
		while for $s_2'>\varepsilon$ it is
		$$
		(s_2'-\varepsilon)x_2(1,s_2') - \int_{\varepsilon}^{s_2'}x_2(1,z)\,dz
		\leq
		(s_2'-\varepsilon)x_2(1,s_2') - (s_2'-\varepsilon)x_2(1,s_2') x_2(1,s_2')
		=0,
		$$
		the inequality holding due to the fact that $x_2(1,z)$ is nonincreasing with respect to $z$.
	\end{proof}
\end{proof}

From  \cref{thm:tradeoff}, \cref{thm:FP} and \cref{thm:vcg}, it is clear that both FP and SP
lie on the boundary of the PoA/PoS feasibility space (see~\cref{fig:pareto_general}).

\section{The Pareto Frontier of Task-Independent Mechanisms}
\label{sec:task-independent}

As we noted in the previous section, both the SP and FP mechanisms, which lie on the
inefficiency boundary (see \cref{fig:pareto_general}), are anonymous
task-independent mechanisms. In this section, we will construct a tighter boundary
on the PoA/PoS trade-off for the class of anonymous task-independent mechanisms.
Furthermore, we will show that this boundary is actually tight, by designing a class
of optimal mechanisms that lie exactly on it, meaning that for each point on the
boundary, there is a mechanism in our class that achieves the corresponding PoA/PoS
trade-off. Thus, this results in a \emph{complete characterization of the Pareto
frontier} between the PoA and the PoS.\footnote{To prevent any potential confusion, we use
the term ``inefficiency boundary'' to refer to the boundary of the feasible region
for the PoA/PoS trade-off, that is defined by some impossibility-type result such as
\cref{thm:task_independent_tradeoff} and we reserve the term ``Pareto frontier'' for
a boundary that can provably not be improved, since there are mechanisms that
achieve the corresponding trade-offs. Intuitively, in our terminology, the
inefficiency boundary is a ``bound'' on the achievable Pareto frontier.} For an
illustration, see \cref{fig:pareto_task_indi}.

\subsection{PoA/PoS Trade-off}
We start with the theorem that gives us the improved boundary on the space of
feasible task-independent and anonymous mechanisms. This is the red line
in~\cref{fig:pareto_task_indi}.
Intuitively, the proof of this theorem is based on the following idea: consider two alternative true cost matrices,
$$
\left(
\begin{array}{c c c c c} 1 & \infty & \cdots & \infty & \infty \\
  \infty & 1 & \ddots & \infty & \infty  \\
   \infty & \infty & \ddots & \ddots & \vdots  \\
  \vdots & \vdots & \ddots & 1 & \infty   \\
  \infty  & \infty  & \ddots & \infty & 1   \\
  \alpha^* & \alpha^* & \cdots & \alpha^* & \alpha^*
\end{array}
\right)
\qquad\text{and}\qquad
\left(
\begin{array}{c c c c c} \alpha & \infty & \cdots & \infty & \infty \\
  \infty & \alpha & \ddots & \infty & \infty  \\
   \infty & \infty & \ddots & \ddots & \vdots  \\
  \vdots & \vdots & \ddots & \alpha & \infty   \\
  \infty  & \infty  & \ddots & \infty & \alpha   \\
  1^* & 1^* & \cdots & 1^* & 1^*
\end{array}
\right),
$$
where $\alpha>1$.
Then, any anonymous mechanism would either (a) have \emph{some} equilibrium where all tasks get allocated to the ``slow'' machine (with running time $\alpha$) on the first instance, or (b) at \emph{all} equilibria it would need to
allocate all tasks to the ``fast'' machine (with running time $1$) on the second instance. Case (a) would result in a ``high'' PoA, while case (b) in a ``high'' PoS.

\begin{theorem}\label{thm:task_independent_tradeoff}
For any task-independent anonymous scheduling mechanism $\mech$ for $n$ machines, and any real $\alpha>1$,
$$
\poa (\mech) < (n-1)\alpha + 1 \quad \then \quad \pos (\mech) \geq \frac{(n-1)}{\alpha}+1.
$$
\end{theorem}
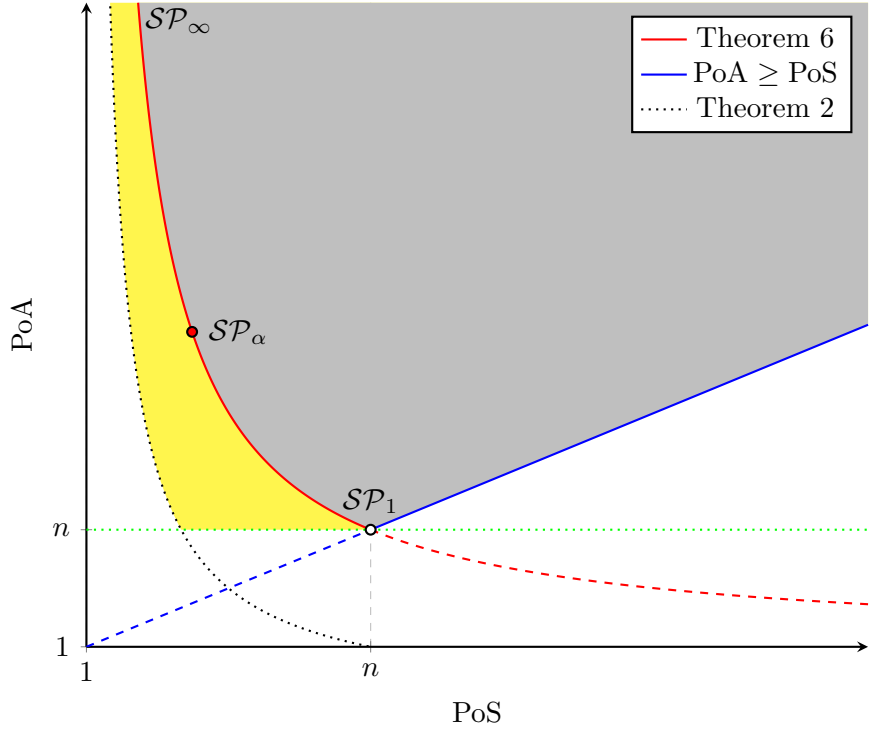
\begin{figure}[t] 
\centering    
\begin{tikzpicture} 
\begin{axis}[
 	scale=1.5,
	xmin=1, xmax=6.5,
	ymin=1, ymax=12,
    axis lines = left,
    ytick={1,3},
    yticklabels={$1$,$n$},
    xtick={1,3},
    xticklabels={$1$,$n$},
    xmajorgrids=true,
    grid style=dashed,
    xlabel = {$\pos$},
    ylabel = {$\poa$},
    thick
]
\path[fill=yellow,opacity=0.7,smooth]
(1.16667,12.) -- (1.17667,11.3208) -- (1.18667,10.7143) -- (1.19667,10.1695) --
(1.20667,9.67742) -- (1.21667,9.23077) -- (1.22667,8.82353) -- (1.23667,8.4507) --
(1.24667,8.10811) -- (1.25667,7.79221) -- (1.26667,7.5) -- (1.27667,7.22892) --
(1.28667,6.97674) -- (1.29667,6.74157) -- (1.30667,6.52174) -- (1.31667,6.31579) --
(1.32667,6.12245) -- (1.33667,5.94059) -- (1.34667,5.76923) -- (1.35667,5.60748) --
(1.36667,5.45455) -- (1.37667,5.30973) -- (1.38667,5.17241) -- (1.39667,5.04202) --
(1.40667,4.91803) -- (1.41667,4.8) -- (1.42667,4.6875) -- (1.43667,4.58015) --
(1.44667,4.47761) -- (1.45667,4.37956) -- (1.46667,4.28571) -- (1.47667,4.1958) --
(1.48667,4.10959) -- (1.49667,4.02685) -- (1.50667,3.94737) -- (1.51667,3.87097) --
(1.52667,3.79747) -- (1.53667,3.72671) -- (1.54667,3.65854) -- (1.55667,3.59281) --
(1.56667,3.52941) -- (1.57667,3.46821) -- (1.58667,3.40909) -- (1.59667,3.35196) --
(1.60667,3.2967) -- (1.61667,3.24324) -- (1.62667,3.19149) -- (1.63667,3.14136) --
(1.64667,3.09278) -- (1.65667,3.04569) -- (1.66667,3.) -- (3,3) --
(\pgfkeysvalueof{/pgfplots/xmax},\pgfkeysvalueof{/pgfplots/xmax}) --
(\pgfkeysvalueof{/pgfplots/xmax},\pgfkeysvalueof{/pgfplots/ymax}) --
(1.667,\pgfkeysvalueof{/pgfplots/ymax}) -- cycle
;
\path[fill=lightgray,opacity=1,smooth]
(1.36364,12.) -- (1.37364,11.7056) -- (1.38364,11.4265) -- (1.39364,11.1617) --
(1.40364,10.9099) -- (1.41364,10.6703) -- (1.42364,10.4421) -- (1.43364,10.2243) --
(1.44364,10.0164) -- (1.45364,9.81764) -- (1.46364,9.62745) -- (1.47364,9.4453) --
(1.48364,9.27068) -- (1.49364,9.10313) -- (1.50364,8.94224) -- (1.51364,8.78761) --
(1.52364,8.63889) -- (1.53364,8.49574) -- (1.54364,8.35786) -- (1.55364,8.22496) --
(1.56364,8.09677) -- (1.57364,7.97306) -- (1.58364,7.85358) -- (1.59364,7.73813) --
(1.60364,7.62651) -- (1.61364,7.51852) -- (1.62364,7.41399) -- (1.63364,7.31277) --
(1.64364,7.21469) -- (1.65364,7.11961) -- (1.66364,7.0274) -- (1.67364,6.93792) --
(1.68364,6.85106) -- (1.69364,6.76671) -- (1.70364,6.68475) -- (1.71364,6.6051) --
(1.72364,6.52764) -- (1.73364,6.45229) -- (1.74364,6.37897) -- (1.75364,6.3076) --
(1.76364,6.2381) -- (1.77364,6.17039) -- (1.78364,6.10441) -- (1.79364,6.04009) --
(1.80364,5.97738) -- (1.81364,5.9162) -- (1.82364,5.85651) -- (1.83364,5.79826) --
(1.84364,5.74138) -- (1.85364,5.68584) -- (1.86364,5.63158) -- (1.87364,5.57856) --
(1.88364,5.52675) -- (1.89364,5.47609) -- (1.90364,5.42656) -- (1.91364,5.37811) --
(1.92364,5.33071) -- (1.93364,5.28432) -- (1.94364,5.23892) -- (1.95364,5.19447) --
(1.96364,5.15094) -- (1.97364,5.10831) -- (1.98364,5.06654) -- (1.99364,5.02562) --
(2.00364,4.98551) -- (2.01364,4.94619) -- (2.02364,4.90764) -- (2.03364,4.86983) --
(2.04364,4.83275) -- (2.05364,4.79638) -- (2.06364,4.76068) -- (2.07364,4.72566) --
(2.08364,4.69128) -- (2.09364,4.65752) -- (2.10364,4.62438) -- (2.11364,4.59184) --
(2.12364,4.55987) -- (2.13364,4.52847) -- (2.14364,4.49762) -- (2.15364,4.4673) --
(2.16364,4.4375) -- (2.17364,4.40821) -- (2.18364,4.37942) -- (2.19364,4.3511) --
(2.20364,4.32326) -- (2.21364,4.29588) -- (2.22364,4.26895) -- (2.23364,4.24245) --
(2.24364,4.21637) -- (2.25364,4.19072) -- (2.26364,4.16547) -- (2.27364,4.14061) --
(2.28364,4.11615) -- (2.29364,4.09206) -- (2.30364,4.06834) -- (2.31364,4.04498) --
(2.32364,4.02198) -- (2.33364,3.99932) -- (2.34364,3.977) -- (2.35364,3.955) --
(2.36364,3.93333) -- (2.37364,3.91198) -- (2.38364,3.89093) -- (2.39364,3.87019) --
(2.40364,3.84974) -- (2.41364,3.82958) -- (2.42364,3.80971) -- (2.43364,3.79011) --
(2.44364,3.77078) -- (2.45364,3.75172) -- (2.46364,3.73292) -- (2.47364,3.71437) --
(2.48364,3.69608) -- (2.49364,3.67803) -- (2.50364,3.66022) -- (2.51364,3.64264) --
(2.52364,3.6253) -- (2.53364,3.60818) -- (2.54364,3.59128) -- (2.55364,3.57461) --
(2.56364,3.55814) -- (2.57364,3.54188) -- (2.58364,3.52583) -- (2.59364,3.50998) --
(2.60364,3.49433) -- (2.61364,3.47887) -- (2.62364,3.46361) -- (2.63364,3.44853) --
(2.64364,3.43363) -- (2.65364,3.41891) -- (2.66364,3.40437) -- (2.67364,3.39001) --
(2.68364,3.37581) -- (2.69364,3.36178) -- (2.70364,3.34792) -- (2.71364,3.33422) --
(2.72364,3.32068) -- (2.73364,3.30729) -- (2.74364,3.29406) -- (2.75364,3.28097) --
(2.76364,3.26804) -- (2.77364,3.25525) -- (2.78364,3.24261) -- (2.79364,3.23011) --
(2.80364,3.21774) -- (2.81364,3.20551) -- (2.82364,3.19342) -- (2.83364,3.18146) --
(2.84364,3.16963) -- (2.85364,3.15792) -- (2.86364,3.14634) -- (2.87364,3.13489) --
(2.88364,3.12355) -- (2.89364,3.11234) -- (2.90364,3.10124) -- (2.91364,3.09026) --
(2.92364,3.0794) -- (2.93364,3.06864) -- (2.94364,3.058) -- (2.95364,3.04746) --
(2.96364,3.03704) -- (2.97364,3.02672) -- (2.98364,3.0165) -- (2.99364,3.00638) --
(3,3) -- (6.5,6.5) -- (6.5,12) -- cycle
;
\addplot [red,domain=1.36364:3,samples=200]
{(3-1)^2/(x-1)+1};
\addplot [red,domain=3:6.5,dashed,samples=200]
{(3-1)^2/(x-1)+1};
\addplot [blue,domain=1:3,dashed,samples=3]
{x};
\addplot [blue,domain=3:6.5,samples=3]
{x};
\addplot [green,domain=1:6.5,dotted,samples=3]
{3};
\addplot [black,domain=1.15:3,dotted,samples=200]
{(3-1)/(x-1)};
\node[circle,draw,fill=white,label=above:{$\mechsp_1$},inner sep=1.3pt] (fp) at (3,3) {};
\node (sp) at (1.65,11.7) {$\mechsp_\infty$};
\node[circle,draw,fill=red,label=east:{$\mechsp_\alpha$},inner sep=1.3pt] (spa) at (1.74364,6.37897) {};
\legend{\cref*{thm:task_independent_tradeoff},,,$\poa\geq\pos$,,\cref*{thm:tradeoff}}
\end{axis}
\end{tikzpicture}
\caption{The inefficiency boundary, for anonymous task-independent mechanisms, given
by \cref{thm:task_independent_tradeoff} (red line). Combined with the global PoA
lower bound of \cref{thm:lower} (green line) and the trivial fact that the PoS is at
most the PoA (blue line), we finally get the grey feasible region. The family of
mechanisms $\mechsp_\alpha$ described in \cref{sec:spa_mechs} lies exactly on this
boundary (red line), thus completely characterizing the \emph{Pareto frontier} in a
smooth way with respect to parameter $\alpha\geq 1$: on its one end ($\alpha=1$) is
the First-Price mechanism $\text{FP}=\mechsp_{1}$ and at the other
($\alpha\to\infty$) the Second-Price mechanism $\text{SP}=\mechsp_{\infty}$. The yellow region
represents the possible trade-off between the Price of Anarchy and the Price of Stability for
mechanisms that are not anonymous or task-independent, as defined by the right boundary in
this figure, and the left boundary of \cref{fig:pareto_general} (shown as a red line there). 
We remark that the ``thickness'' of this region is mainly due to mechanisms which are not task-independent; as we show in \cref{thm:sqrt-tradeoff}, the corresponding boundary line for
mechanisms which are task-independent but not anonymous is much closer to the red line in the figure.}
\label{fig:pareto_task_indi}
\end{figure}
\begin{proof}
Fix a mechanism $\mech$ on $n\geq 2$ machines 
and $m = \max\{2,2n-3\}$ tasks, that allocates each task $j$
independently by running an anonymous single-task mechanism $\AA_j$. Each such
mechanism $\AA_j$ takes as input a declared cost vector $\vecc s^j= (s_{1,j},
\dots, s_{n,j})$  by the machines, where $s_{i,j}$ is the report of machine $i$ for
task $j$, $i=1,\dots,n$. Besides the reports, recall that there is also an
underlying \emph{true} cost vector $\vecc t^j=(t_{1,j},\dots,t_{n,j})$. Also, fix a
parameter $\alpha>1$ and assume that $\poa (\mech) < (n-1)\alpha + 1$.

We are particularly interested in true cost vectors that have a specific structure,
namely being permutations of $(1,\alpha,\infty,\dots,\infty)$. We will call such
cost vectors \emph{canonical} for the remainder of this proof. Formally, $\vecc t^j$
is canonical if:
\begin{itemize}
\item there is a unique machine $i'$ such that $t_{i',j}=1$ (machine $i'$ will be
  called \emph{fast} for task $j$),
\item there is a unique machine $i''\neq i'$ such that $t_{i'',j}=\alpha$ (machine
$n$ will be called \emph{slow} for task $j$),
\item all other machines $i\neq i',i''$ have arbitrarily high, pairwise distinct,
processing times for task $j$, that for simplicity we'll denote with
$t_{i,j}=\infty$ (these machines will be called \emph{dummy}).
\end{itemize} Notice that, since $\mech$ has a bounded PoA, none of its component
mechanisms $\AA_j$ can have an equilibrium,\footnote{In multiple points throughout
this proof we will silently be using the fact that a profile of reports $\vecc s$ is
an equilibrium of $\mech$ with respect to a true profile $\vecc t$, if and only if,
for all tasks $j$, $\vecc s^j$ is an equilibrium of $\AA_j$ with respect to $\vecc
t^j$; this is an immediate consequence of task-independence
(see~\cref{def:task-ind}).} on any true canonical cost vector $\vecc t^j$, that
allocates task $j$ to a dummy machine.

The following definition will be helpful for our exposition in the rest of the
proof:
\begin{definition}[Well-behaved tasks]\label{def:well-behaved} A task $j$ will be
called \emph{well-behaved} if, for all true canonical cost vectors $\vecc t^j$,
mechanism $\AA_j$ allocates task $j$ to the fast machine on \emph{all} equilibria
$\vecc s^j$.
\end{definition} It turns out that, due to anonymity, a much weaker condition is
actually enough in order to establish that a task is well-behaved:
\begin{claim}\label{claim:well-behaved_weaker} A task is well-behaved if there
exists a true canonical cost vector with respect to which all equilibria assign it
to the fast machine.
\end{claim}
\begin{proof} Assume that there exists a canonical cost vector $\vecc t^j$ such
that, for all equilibria $\vecc s^j$ of $\vecc t^j$, mechanism $\AA_j$ assigns task
$j$ to its fast machine. For a contradiction, assume that there also exists a canonical cost
vector $\tilde{\vecc t}^j$ and an equilibrium $\tilde{\vecc s}^j$ under which
$\AA_j$ assigns the task to the slow machine (recall that the task cannot be
assigned to a dummy machine). Then, since  $\vecc t^j$ is a permutation of
$\tilde{\vecc t}^j$ and canonical vectors do not have ties, due to anonymity
(see~\cref{def:anonymity}) there has to exist an equilibrium $\hat{\vecc s}^j$ with
respect to $\vecc t^j$, under which the task is given to the slow machine; this is a
contradiction.
\end{proof}

\begin{claim}\label{claim:fixed_well-behaved} Mechanism $\mech$ has at least $n-1$
well-behaved tasks.
\end{claim}
\begin{proof} Let $\vecc t$ be a true instance whose task cost vectors are all
canonical. For example, let
\[
\vecc t =
\underbrace{\begin{pmatrix} 1 & 1 & \cdots & 1 \\
\alpha & \alpha & \cdots & \alpha \\
\infty & \infty & \cdots & \infty \\
\vdots & \vdots &  & \vdots \\
\infty & \infty & \cdots & \infty \\
\end{pmatrix}}_{2n-3}.
\] 
Due to~\cref{claim:well-behaved_weaker}, it is enough to show that there exists a
fixed set of (at least) $n-1$ tasks which, under \emph{all} equilibria of $\vecc t$,
they all get allocated to their fast machines. Formally, to get to a contradiction,
assume that for any subset of tasks $\bar{J}\subseteq[m]$ with $|\bar{J}|=n-1$, there exists a
task $j=\mathfrak{j}(\bar{J})\in \bar{J}$ and an equilibrium $\vecc s^j$ of $\vecc t^j$ on which
$\AA_j$ allocates task $j$ to the slow machine. Then, since instance $\vecc t$ has
at least $m\geq 2n-3$ tasks, we can apply this property repeatedly in order to get a
sequence of $n-1$ tasks
$$
j_1=\mathfrak{j}\left([n-1]\right),\qquad 
j_2=\mathfrak{j}\left([n]\setminus j_1\right),\qquad \dots, \qquad
j_{n-1}= \mathfrak{j}\left([2n-3]\setminus\{j_1, j_2, \dots, j_{n-2}\}\right)
$$
and corresponding equilibria $\vecc s^{j_1},\vecc s^{j_2},\dots,\vecc s^{j_{n-1}}$
(with respect to true canonical cost vectors $\vecc t^{j_1},\vecc
t^{j_2},\dots,\vecc t^{j_{n-1}}$, respectively) with the property that, for all
$\ell=1,\dots,n-1$, mechanism $\AA_{j_\ell}$ allocates task $j_{\ell}$ to the slow
machine.

Without loss, assuming for the rest of this claim's proof that
$\sset{j_1,j_2,\dots,j_{n-1}}=[n-1]$, we consider the following new profile of true
costs
$$
\tilde{\vecc t}=
 \left(
\begin{array}{c c c c c | c} 1 & \infty & \infty & \cdots & \infty &
\multirow{5}{*}{\text{\smash{\raisebox{-1.2ex}{\huge 0}}}}\\
  \infty & 1 & \infty & \cdots & \infty & \\
  \vdots & \ddots & \ddots & \ddots & \vdots &  \\
  \infty  & \cdots  & \infty & 1 & \infty &  \\
  \alpha & \alpha & \cdots & \alpha & 1 &
\end{array}
\right),
$$
where tasks $j=1,\dots,n-1$ have canonical cost vectors with the fast machine at the
diagonal and the slow one always being $i=n$; task $j=n$ is arbitrarily slow on all
machines except $i=n$, for which it has a cost of $1$; and all remaining tasks
$j\geq n+1$ have been rendered essentially irrelevant by setting their running times
to $0$ for all machines.

Since all canonical vectors $\tilde{\vecc t}^j$, $j=1,\dots,n-1$, of the new true
profile $\tilde{\vecc t}$ are permutations of the ones in the original true profile
$\vecc{t}$, and additionally we established that there exist equilibria $\vecc
s^{j}$ with respect to $\vecc t$ that assign all these tasks to their slow machines,
then due to anonymity (see~\cref{def:anonymity}) there must also exist an
equilibrium $\tilde{\vecc s}$ of mechanism $\mech$ (with respect to $\tilde{\vecc
t}$) on which all tasks $j=1,\dots,n-1$ are given to the machine with cost $\alpha$.
Furthermore, clearly task $j=n$ needs to be allocated to machine $i=n$ as well on
$\tilde{\vecc s}$, since she is the only one with bounded running time (recall that
$\mech$ has a bounded PoA).

Summarizing, equilibrium $\tilde{\vecc{s}}$ assigns all tasks to the last machine, for a
makespan of $(n-1)\cdot \alpha+1$. On the other hand, the diagonal allocation on
$\tilde{\vecc t}$, i.e.\ giving each task to the machine with cost $1$ would have
given a makespan of $1$. This results in a PoA bound of at least
$
\poa(\mech)\geq (n-1)\alpha+1,
$
which contradicts our initial assumptions about mechanism $\mech$.
\end{proof}

In light of~\cref{claim:fixed_well-behaved}, without loss let's assume from now on
that the first $n-1$ tasks are well-behaved, and consider the following profile of
true costs
$$
\hat{\vecc t}=
 \left(
\begin{array}{c c c c c | c}
\alpha & \infty & \infty & \cdots & \infty &
\multirow{5}{*}{\text{\smash{\raisebox{-1.2ex}{\huge 0}}}}\\
  \infty & \alpha & \infty & \cdots & \infty & \\
  \vdots & \ddots & \ddots & \ddots & \vdots &  \\
  \infty  & \cdots  & \infty & \alpha & \infty &  \\ 1 & 1 & \cdots & 1 & \alpha &
\end{array}
\right),
$$
which, in a similar way to that in the proof of~\cref{claim:fixed_well-behaved} for
profile $\tilde{\vecc{t}}$, we derived by permuting accordingly the canonical vectors
of the first $n-1$ tasks of $\vecc t$, setting the cost of task $j=n$ to be bounded
only for the last machine and, finally, rendering all tasks $j>n$ irrelevant.

Fix now \emph{any} equilibrium $\hat{\vecc{s}}=(\hat{\vecc{s}}_1,\dots,\hat{\vecc
s}_{n-1},\hat{\vecc{s}}_n)$ of $\mech$ under $\hat{\vecc t}$. 
Since tasks $j=1,\dots,n-1$ are well-behaved (see~\cref{def:well-behaved}), under
$\hat{\vecc{s}}$ mechanism $\mech$ has to assign them all to the machine with
running time $1$. Task $j=n$ needs to be allocated to machine $i=n$ as well (since
$\poa(\mech)$ is bounded). Thus, equilibrium $\hat{\vecc{s}}$ results in a makespan
of $(n-1)\cdot 1+\alpha$, while the diagonal allocation on $\hat{\vecc t}$ has a
makespan of $\alpha$. This results in a bound of
$
\pos(\mech)\geq \frac{n-1}{\alpha}+1,
$
concluding the proof of the theorem.
\end{proof}

\subsection{Optimal Mechanisms on the Pareto Frontier}
\label{sec:spa_mechs}

Next, we will design a class of mechanisms, parameterized by a quantity $\alpha$
that will populate, in a smooth way, the boundary given
by~\cref{thm:task_independent_tradeoff}. Thus, these mechanisms achieve trade-offs
that lie on the Pareto frontier of inefficiency for the class of task-independent
and anonymous mechanisms.

\begin{definition}[Second-Price mechanism with $\alpha$-relative reserve price ($\mechsp_\alpha$)]
\label{def:spa}
For $\alpha\geq 1$, $\mechsp_\alpha$ is the task-independent mechanism that, for
each task $j$: finds a machine $k \in \arg\min_{i \in N} s_{i,j}$ and sets a
\emph{reserve price} at $r=\alpha \cdot s_{k,j}$; assigns the task to the fastest
machine $\iota(j)\in \argmin_{i \in N} s_{i,j}$ (breaking ties-arbitrarily); pays
machine $\iota(j)$ the amount $\min \{\min_{i \in N \setminus
\{\iota(j)\}}s_{i,j},r\}$; pays nothing to the remaining machines $N\setminus
\iota(j)$.
\end{definition}
Informally, for each task $j$, the mechanism sets a reserve price which is $\alpha$
times larger than the smallest declared processing time, allocates the task to the
fastest machine (according to the declarations) and pays the machine the minimum of
the second-smallest declared processing time and the reserve price. What this
mechanism achieves in terms of the equilibria that it induces is the following:
assume that we create a \emph{bucket} of tasks with \emph{true} processing
times at most $\alpha$ times larger than the smallest \emph{true} processing time.
Then, in every equilibrium of the mechanism, task $j$ is allocated to some machine
in the bucket and moreover, for any machine in the bucket, there exists some
equilibrium under which $\mechsp_\alpha$ allocates the task to that machine. This is captured
formally by the following two lemmas. Referencing our discussion in
\cref{sec:extema}, we remark that in the case of $\text{FP}=\mechsp_1$, the tasks can 
only be assigned to the fastest machine(s), and in the case of
$\text{SP}=\mechsp_\infty$, the bucket contains the whole set of machines.
\begin{lemma}[``Nothing outside the bucket'']
\label{lem:poa-spr}
In any equilibrium of $\mechsp_\alpha$, any task $j$ can only be assigned to a
machine with processing time at most $\alpha \cdot \min_{i \in N} t_{i,j}$.
\end{lemma}
\begin{proof}
Since $\mechsp_\alpha$ is task-independent, it suffices to consider the equilibria
of a single component $\mechsp_\alpha^j$, corresponding to task $j$. Let
${\ss}^j$ be such an equilibrium and without loss of generality, assume that
$t_{1,j} \in \arg\min_{i \in N} t_{i,j}$, i.e., machine $1$ is the fastest machine
according to the true processing times. Assume by contradiction that in
${\ss}^j$, some task $\ell$ with real processing time $t_{\ell, j} > \alpha
\cdot t_{1,j}$ is allocated task $j$ and let $s_{\ell, j}$ be its report. Since
machine $\ell$ receives the task, it obviously holds that $s_{\ell, j} \in
\arg\min_{i \in N} s_{i,j}$. We will consider two cases.

  \underline{Case 1:} $s_{\ell, j} \leq t_{1,j}$. In this case, the reserve price is set at
  $r = \alpha \cdot s_{\ell, j} \leq \alpha \cdot t_{1,j}$ and machine $\ell$ receives
  a payment of at most $\alpha \cdot t_{1,j}$. By assumption however, its true
  processing time is larger than $\alpha \cdot t_{i,j}$ and therefore the machine has
  negative utility. By deviating to telling the truth, the machine can obtain a
  nonnegative utility, contradicting the fact that ${\ss}^j$ is an
  equilibrium.

  \underline{Case 2:} $s_{\ell, j} > t_{1,j}$. In this case, machine $1$ has $0$ utility,
  since she is not allocated the task and she is not paid anything. However, if
  machine $1$ deviates to telling the truth, (i.e., if she deviates to
  $s'_{i,j}=t_{1,j}$), then, since $s_{\ell, j} = \min_{i \in N} s_{i,j}$ by assumption,
  the machine will now win the task (i.e., $x_{1j}=1$) and will receive a payment of
  $s_{\ell, j} > t_{1,j}$, obtaining strictly positive utility. Again, this contradicts
  the fact that ${\ss}^j$ is an equilibrium.

  In any case, ${\ss}^j$ can not be an equilibrium in which task $j$ is
  allocated to some machine with processing time larger than $\alpha \cdot \min_{i
  \in N} t_{i,j}$.
\end{proof}

\begin{lemma}[``Everything inside the bucket'']
\label{lem:pos-spr}
For every input profile ${\t}$ and any $\alpha >1$, there exist equilibria of $\mechsp_\alpha$
such that for every task $j$, every machine with processing time at most $\alpha
\cdot \min_{i \in N} t_{i,j}$ can be allocated task $j$.
\end{lemma}
\begin{proof}
Again, since $\mechsp_\alpha$ is task-independent, it suffices to consider the
equilibria of a single component $\mechsp_\alpha^j$, corresponding to task $j$.
Given an input profile ${\t}$, let $J_\alpha$ be the set of machines $i'$ such
that $t_{i',j} \leq \alpha \cdot \min_{i \in N} t_{i,j}$ and let $t_f = \min_{i \in N}
t_{i,j}$ be the processing time of the fastest machine for task $j$. Let $k$ be any
machine in $J_\alpha$; we will consider two strategy profiles (or rather their restrictions
to the $j$-th component), depending on whether $t_{k,j} > t_f$ or $t_{k,j} = t_f$.

  \underline{Case 1:} $t_{k,j} > t_f$. Then, consider the strategy profile ${\ss}^j$
  such that $s_{k,j} = t_f$ and $s_{i,j} = t_{k,j}$ for any machine $i \in J
  \setminus \{k\}$. We will prove that ${\ss}^j$ is an equilibrium, by considering
  possible deviations of machine $k$, and the remaining machines in $J \setminus
  \{k\}$ separately.
  \begin{enumerate}
  	\item[a)] The current utility of machine $k$ is $0$, since it receives the task
	for which the true processing time is $t_{k,j}$ and receives a payment of
	$t_{k,j}$. Note that the reserve price is set to at most $t_{k,j}$, since $t_{k,j}
	\leq \alpha \cdot t_f$ by assumption, and $s_{k,j} = t_f$. Since there exist
	machines with reported processing times at $t_{k,j}$, machine $k$ can not obtain
	positive utility by any deviation, even if it increases the reserve price
	while still winning the task.
	\item[b)] Consider any machine $i \in J \setminus \{k\}$. Since machine $i$ is
	not winning the task, its utility is $0$, and the only way to possibly obtain a
	positive utility is by forcing an allocation in which it wins the task. For this
	to be possible, it has to deviate to $s'_{i,j} \leq t_f = s_{k,j}$, as otherwise
	machine $k$ would still be the winner. In that case however, the payment of
	machine $i$ will be at most $t_f$ and by assumption, we know that $t_{i,j} \geq
	t_f$ for any task $j \in J$. Therefore, the deviation results in a utility of at
  	most $0$ for machine $i$ and it is not a beneficial deviation.
  \end{enumerate}
  
 \underline{Case 2:} $t_{k,j} = t_f$. Then, consider the strategy profile ${\ss}^j$
 such that $s_{k,j} = t_f$ and $s_{i,j} = t_f + \varepsilon$, where $t_f+\varepsilon <
 \alpha \cdot t_f$; this is possible since $\alpha >1$ and the continuity of the
 strategy space. Again, we consider possible deviations of machine $k$, and the
 remaining machines separately.
\begin{enumerate}
	\item[a)] The current utility of machine $k$ is $\varepsilon$, since it receives
	the task for which the true processing time is $t_{k,j}=t_f$ and receives a
	payment of $t_f+\varepsilon$, the second smallest reported processing time (which
	is also smaller than the reserve price by the choice of $\varepsilon$). In order
	for machine $k$ to still receive the task, it has to use some strategy $s_{k,j}'
	\leq t_f + \varepsilon$, but any such strategy can not affect the payment that it
	receives. Therefore the machine does not have a beneficial deviation.
	\item[b)] The argument in this case is identical to Case 1b above.
\end{enumerate}
\end{proof}

\begin{theorem}\label{thm:sp_alpha_poa}
The Price of Anarchy of $\mechsp_\alpha$ on $n$ machines is at most $(n-1) \alpha +1$.
\end{theorem}

\begin{proof}
Fix some underlying $n\times m$ true cost matrix $\vecc t$ and a parameter $\alpha
>1$. Fix also an optimal (makespan-minimizing) allocation $\opt$ of $\vecc t$ and a
(pure) Nash equilibrium $\vecc s$ of $\mechsp_\alpha$ under true costs $\vecc t$.
For any task $j=1,\dots,m$, let $\iota^{*} (j)$ and $\iota(j)$ denote the machine
that gets task $j$ at $\opt$ and $\mechsp_\alpha(\vecc s)$, respectively. Also, let
$K_i^{*}$, $K_i$ denote the corresponding machine loads and $J_i^{*}$, $J_i$ the
sets of assigned tasks; that is, for $i=1,\dots,n$, we define
$$
K_i^{*} \equiv \sum_{j\in J_i^{*}} t_{i,j}
\qquad\text{and}\qquad
K_i \equiv \sum_{j\in J_i} t_{i,j},
$$
where
$$
J_i^{*}=\ssets{j\in J\fwh{\iota^{*}(j)=i}}
\qquad\text{and}\qquad
J_i=\ssets{j\in J\fwh{\iota(j)=i}}.
$$
Finally, it is without loss of generality to assume that $K_1\geq K_2 \geq \dots
\geq K_{n}$, so that the makespan of $\mechsp_\alpha$ on $\vecc s$ is
\begin{align*}
K_1=\sum_{j\in J_1} t_{1,j}
=\sum_{j\in J_1\inters J_1^{*}} t_{1,j} + \sum_{j\in J_1\setminus J_1^{*}} t_{1,j}
\leq \sum_{j\in J_1^{*}} t_{1,j} + \alpha\sum_{j\in J\setminus J_1^{*}} t_{\iota^{*}(j),j},
\end{align*}
the last inequality holding since, for any task $j$, $t_{\iota(j),j}\leq \alpha
\cdot t_{\iota^{*}(j),j}$ (due to \cref{lem:poa-spr}). Thus, we can bound our
mechanism's makespan by
$$
K_1\leq K_1^{*}+\alpha\sum_{i=2}^{n} K_i^{*} \leq K_1^{*}+\alpha (n-1)\max_{i=2,\dots,n}K_i^{*}.
$$
Putting everything together, and denoting for simplicity $x=K_1^{*}$ and
$y=\max_{i=2,\dots,n}K_i^{*}$, the PoA of $\mechsp_\alpha$ is finally upper bounded
by
\begin{equation*}
\poa(\mechsp_\alpha)
\leq \frac{x+\alpha (n-1)y}{\max_{i=1,\dots,n} K_i^{*}}
= \frac{x+\alpha (n-1)y}{\max\sset{x,y}}
\leq \alpha (n-1)+1,
\end{equation*}
the last step coming from applying \cref{lemma:tech_1} (see appendix) with $\beta=\alpha (n-1)$ and $\gamma =1$.
\end{proof}

\begin{theorem}\label{thm:sp_alpha_pos}
The Price of Stability of $\mechsp_\alpha$ on $n$ machines is at most $\frac{n-1}{\alpha}+1$.
\end{theorem}

\begin{proof}
Fix some underlying $n\times m$ true cost matrix $\vecc t$,
an optimal (makespan-minimizing) allocation $\opt$ of $\vecc t$
and a parameter $\alpha >1$.
Also, let $\iota^{*} (j)$ denote the machine that gets task $j$ at $\opt$.

We partition the set of tasks $J=\ssets{1,2,\dots,m}$ into two sets
$J_{\text{small}}$ and $J_{\text{large}}$, based on their processing time under
$\opt$. More specifically, we define:
$$
J_{\text{small}} \equiv \ssets{j\in J \fwhs{ t_{\iota^{*}(j),j} \leq \alpha\cdot\min_{i\in N} t_{i,j}}}
\quad
\text{and}
\quad
J_{\text{large}} \equiv J\setminus J_{\text{small}}= \ssets{j\in J \fwhs{ t_{\iota^{*}(j),j} > \alpha\cdot\min_{i\in N} t_{i,j}}}.
$$
Intuitively, in light of \cref{lem:poa-spr,lem:pos-spr}, we can think of
$J_{\text{small}}$ as containing all the tasks that $\mechsp_\alpha$ can allocate to
the same machine as $\opt$ in some equilibrium and $J_\text{large}$ containing the
tasks for which this is not possible, but which nevertheless end up at a machine
that runs them faster than in $\opt$.

Now consider the pure Nash equilibrium $\ss$ of $\mechsp_\alpha$ (with respect to
true running times $\vecc{t}$) that allocates all small tasks at the same machine as
$\opt$, and all large tasks myopically to the fastest machine for that task. More
precisely, if $\iota(j)$ denotes the machine that gets task $j$ under
$\mechsp_\alpha(\ss)$, we set $\iota(j) = \iota^{*}(j)$ for all $j\in
J_{\text{small}}$ and $\iota(j)\in \argmin_{i\in N} t_{i,j}$ for $j\in
J_{\text{large}}$. The fact that such an equilibrium indeed exists, is a consequence
of \cref{lem:pos-spr}.

Let $K_i$ denote the load of machine $i$ after allocating only the small tasks
according to $\opt$ (and thus, also according to the equilibrium $\ss$ of
$\mechsp_\alpha$). Let $L_i^{*}$, $L_i$ denote the load of machine $i$ after
allocating only the large tasks according to $\opt$ and $\ss$, respectively. That
is, we formally define:
$$
K_i \equiv \negthickspace\negthickspace\sum_{j\in J_{\text{small}} :\iota(j)=i}\negthickspace\negthickspace t_{i,j},
\qquad
L_i^{*} \equiv \negthickspace\negthickspace \sum_{j\in J_{\text{large}} :\iota^{*}(j)=i} \negthickspace\negthickspace t_{i,j}
\qquad\text{and}\qquad
L_i \equiv \negthickspace\negthickspace \sum_{j\in J_{\text{large}} :\iota(j)=i} \negthickspace\negthickspace t_{i,j}.
$$
Without loss, let's assume that $K_1\geq K_2\geq\dots\geq K_n$.
Then, the makespan of $\opt$ is
$$
\max_{i\in N} (K_i+L_i^{*})
\geq
\max\sset{K_1+L_1^{*},\max_{i=2,\dots,n} L_i^{*}}
\geq
\max\sset{K_1+L_1^{*},\frac{1}{n-1}\sum_{i=2,\dots,n} L_i^{*}},
$$
while that of $\ss$ can be upper bounded by
$$
\max_{i\in N}(K_i+L_i)
\leq
\max_{i\in N} K_i + \max_{i\in N} L_i
\leq K_1+\sum_{i=1}^n L_i
\leq K_1+\frac{1}{\alpha}\sum_{i=1}^n L_i^{*}
\leq K_1+L_1^{*}+\frac{1}{\alpha}\sum_{i=2}^n L_i^{*},
$$
where the second to last inequality holds due to the fact that for large tasks
$$
\sum_{i=1}^n L_i
=\sum_{j\in J_{\text{large}}} t_{\iota(j),j}
=\sum_{j\in J_{\text{large}}} \min_{i\in N} t_{i,j}
\leq \sum_{j\in J_{\text{large}}} \frac{1}{\alpha} t_{\iota^{*}(j),j}
= \frac{1}{\alpha}\sum_{i=1}^n L_i^{*},
$$
and the last one holds due to $\alpha \geq 1$.

Putting everything together, and denoting for simplicity $x=K_1+L_1^{*}$ and $y=\sum_{i=2}^n L_i^{*}$, we have that
$$
\pos(\mechsp_{\alpha}) \leq \frac{x+\frac{1}{\alpha} y}{\max\sset{x,\frac{1}{n-1}y}} \leq \frac{n-1}{\alpha}+1,
$$
the last inequality holding by applying \cref{lemma:tech_1} with $\beta=\frac{1}{\alpha}$ and $\gamma=\frac{1}{n-1}$.
\end{proof}

\section{Discussion and Future Directions}
\label{sec:discussion}

In this section, we discuss some implications of our approach, as well as directions
for future work. On a general level, one could follow our agenda of studying the
inefficiency trade-off between the Price of Anarchy and the Price of Stability for
many other problems in algorithmic mechanism design, such as auctions
\cite{Lucier2013,syrgkanis2013composable}, machine scheduling without money
\cite{K14,GKK16}, or resource allocation \cite{christodoulou2016social}, to name a
few, for which the two inefficiency notions have already been studied separately.

In terms of the strategic scheduling setting, our work gives rise to a plethora of
intriguing questions for future work, both on a technical and a conceptual level,
which we highlight below in more detail.

\subsection{General Mechanisms} 
\label{sec:discussion_general_mechs}

Why did we focus on task-independent mechanisms for our tight frontier result, since
they are seemingly not a good fit for makespan minimization? First of all, the
latter statement is true for the \emph{algorithmic} version of the problem, but not
necessarily true for the \emph{strategic} version. It is not only conceivable that
task-independent mechanisms, despite their ``naive'' allocation rules, can induce
equilibria in which the allocation is actually quite efficient, but this actually
happens, as evidenced by the Price of Stability of the Second Price Mechanism. Also,
task-independent mechanisms are much more amenable to an equilibrium analysis,
because each task induces a separate games between the machines. Showing the
limitations of this class is quite important, because it is not a priori clear that
they could not achieve the best possible trade-offs suggested by
\cref{thm:tradeoff}.

The major open question associated with our work is whether there exists a mechanism
that achieves a better trade-off than that of \cref{thm:task_independent_tradeoff},
or in other words,
\begin{quote}
\emph{``Is the yellow region of \cref{fig:pareto_task_indi} empty or not?''}
\end{quote} 
If such a mechanism exists, it will most probably \emph{not} be
task-independent\footnote{See also the discussion in \cref{sec:discussion_anonymity}.}; this is
somewhat reminiscent of the state-of-the-art results in truthful machine scheduling,
where the best possible mechanisms (with respect to the approximation ratio, see
\cref{sec:solutionconcepts}) for several variants of the problem are in fact
task-independent \cite{NR07,christodoulou2010mechanism} and whether a better
mechanism that is not task-independent exists is a prominent open question.

While in the case of truthful mechanisms, the general consensus\footnote{This is
	because of the result of \citet{ashlagi2012optimal} for anonymous mechanisms and
	due to personal communications with authors of central papers in the field.}
	seems to be that the best achievable mechanisms will in fact eventually proven
	to be task-independent, the situation in the strategic version of the problem
	might be quite different. This is because we have \emph{all} possible mechanisms
	at our disposal and it is more conceivable that some allocation rule, tied with
	some appropriate payment function could potentially outperform the trade-off
	bounds of \cref{thm:task_independent_tradeoff}. This seems, however, like a
	quite challenging task; to offer some intuition, we remark the following about
	the design of general mechanisms for the problem.

The most natural idea is perhaps to use a known algorithm for unrelated machine
scheduling~\citep{ibarra1977heuristic,Davis1981,LST90} such as the greedy allocation
algorithm\footnote{The algorithm is referred to as ``Algorithm D'' in
\cite{ibarra1977heuristic}.} of \citet{ibarra1977heuristic}, or even the
makespan-optimal algorithm and couple them with a ``get-paid-your-load'' payment
function, where each machine receives a monetary compensation equal to the sum of
the reported processing times for the tasks that she gets assigned. This is
essentially the generalization of the payment rule of the First-Price mechanism for
more general allocation rules. The hope is that by virtue of having a more efficient
allocation rule, the resulting mechanism will always have equilibria with small
makespan (good PoS) while never having equilibria with very large makespan (good
PoA). Unfortunately, one can show that for a large class of such allocation
algorithms (which includes all the aforementioned algorithms that have been proposed
in the literature for the classical unrelated machine scheduling problem), the Price
of Anarchy of the resulting mechanism will be unbounded.

Overall, for a mechanism to lie in the yellow region of~\cref{fig:pareto_task_indi},
it seems imperative that it will need to employ some more complicated payment
function, which will ``guide'' the agents towards the desired equilibria, rather
than simply attempt to implement a better allocation rule with known payment
structures.

\subsection{The Role of Anonymity}
\label{sec:discussion_anonymity}

In the proof of~\cref{thm:task_independent_tradeoff}, we used the fact that the
mechanism in question is anonymous. In many cases in the general related literature,
this assumption is without loss of generality, as the best possible mechanisms with
respect to an objective are anonymous. In the literature of the unrelated machine
scheduling problem in algorithmic mechanism design however, understanding the role
of anonymity is a long-standing open problem. In particular, while all the known
mechanisms for several variants of the problem are anonymous, the best-known lower
bounds for anonymous and non-anonymous mechanisms are strikingly different, from,
specifically $n$ in the former case (given by \cite{ashlagi2012optimal} and matching
the best-known upper bound of \cite{NR07}) and $2.75$ in the latter case (given by
\cite{giannakopoulos2020new}). These challenges are inherited in the strategic setting
as well. 
Nevertheless, for general (not necessarily anonymous) task-independent mechanisms, we can still
show the following inefficiency trade-off. Its proof can be found in
\cref{sec:app-sqrt2}.

\begin{theorem}\label{thm:sqrt-tradeoff}
	For any task-independent scheduling mechanism $\mech$ for $n$ machines, and real
	$\alpha>1$,
	$$
	\poa (\mech) < (n-1)\frac{\alpha}{\sqrt{2}}+1 \quad \then \quad \pos (\mech)
	\geq \frac{(n-1)}{\alpha \sqrt{2}}+1.
	$$
\end{theorem}
The above is a strict improvement with respect to the general boundary given by
\cref{thm:tradeoff}, but still does not quite match the Pareto frontier that we
proved for anonymous mechanisms. In other words, the inefficiency boundary given by
\cref{thm:sqrt-tradeoff} lies strictly within the yellow area in
\cref{fig:pareto_task_indi}. Obtaining tight bounds is an interesting open problem.

\subsection{Equilibrium Notion Considerations}
\label{sec:disc-equilibrium}
In this paper, we study the set of all possible equilibria of mechanisms for the
problem, which may include equilibria which are \emph{weakly dominated} (e.g.\
see~\citep[Sec.~1.8]{Myerson1997a}), i.e., the agents could use a different strategy
instead of their equilibrium strategy and obtain the same utilities, regardless of
the reports of the other agents. These type of equilibria are known to exist in the
Second-Price mechanism and our class of mechanisms $\mechsp_\alpha$ also exhibits
such equilibria. In order to quantify these type of equilibria  in terms of the
agents' aversion to risk, \citet{babaioff2014efficiency} defined the notion of
\emph{exposure factor}, which measures the amount of risk that an agent is willing
to expose herself to, when best-responding. In the terminology of our setting, the
exposure $\gamma$ of strategy $s_i$ is such that
\[
p_i (s_i,s_{-i}) \geq (1+\gamma) \sum_{i \in S_i} t_i,
\]
where $s_{-i}$ is any vector of strategies of the other machines and $S_i$ is the
set of tasks assigned to machine $i$ under $\vecc{x}(s_i, s_{-i})$. In simple
words, $\gamma$ is used to quantify how much extra cost agent $i$ could possibly
experience, if all other agents coordinated to a strategy that is the worst possible
for the agent. Then, \citet{babaioff2014efficiency} proceed to define the set
$\mathcal{Q}_t^\gamma$ as the set of all equilibria that consist only of strategies
with exposure factor at most $\gamma$ and the corresponding notion of the Price of
Anarchy with the respect to this equilibrium set. Given a parameter $\alpha$,
mechanism $\mechsp_\alpha$ can be seen as achieving a PoS guarantee even with
respect to the set $Q_t^{\alpha-1}$ of equilibria consisting of strategies of
exposure at most $\alpha-1$. Conceptually, even if one is only willing to accept a
certain level of risk exposure, an
appropriate $\beta$ can be chosen and the corresponding mechanisms $\mechsp_\alpha$
for $\alpha \leq \beta$, will lie on the inefficiency boundary, even if the solution
concept is the $\gamma$-exposure Nash equilibrium.

Whether there exists a mechanism that can match,
in \emph{undominated} Nash equilibria, the guarantees of $\mechsp_\alpha$, is an interesting
open question; we believe though that this is rather unlikely. That being said,
proving impossibility results for general classes of mechanisms seems quite
challenging, as the property of not being dominated does not convey much information
from a technical standpoint and in particular has different implications for
different mechanisms. This is in contrast with the more standard approach in
auctions, where \emph{specific} mechanisms have been studied, for which undominated
strategies imply a very handy non-overbidding property, e.g., see
\cite{feldman2016correlated,markakis2015uniform}.

\subsection{Computational Considerations}
\label{sec:discussion_computation}

As we mentioned earlier, the mechanisms that we construct in this paper
(see~\cref{def:spa}) are rather simple and run in polynomial time, and this is
actually the case for all known mechanisms for the truthful scheduling problem as
well. It would be interesting to investigate whether adding \emph{computational
efficiency} as a desirable property of the mechanisms in question can have any
implications on the inefficiency boundary. For truthful scheduling, this is unlikely
to be an issue, since the constraint of truthfulness itself typically leads to
rather simple mechanisms which are easily seen to be efficient. As we hinted
in~\cref{sec:discussion_general_mechs} however, it is conceivable that from the
space of all possible mechanisms that we can use, the best one might employ an
allocation algorithm that is computationally intractable. Concretely, it could be
possible that the makespan-optimal algorithm (which does not run in polynomial time,
unless P=NP~\citep{LST90}) can be coupled with an appropriate payment function to
achieve a better trade-off guarantee. From this discussion, we deduce the following,
very interesting question:
\begin{quote}
	\emph{``If we add computational efficiency as a constraint, can we prove a
	stronger inefficiency boundary than that of~\cref{fig:pareto_general}?.''}
\end{quote}
Conceptually, the question above regards whether there is a fundamental connection
between the running time of the allocation rule and the PoA/PoS trade-off that can
be explored via a corresponding inefficiency boundary.

\paragraph*{Acknowledgements} 
We are grateful to the anonymous reviewers of the journal version of this paper for
their careful reading of our manuscript and for their valuable comments that helped
us to simplify and improve the presentation at key points of our work. We also
thank Elias Koutsoupias, Maria Kyropoulou and Diogo Poças for useful discussions.

\bibliographystyle{abbrvnat}
\bibliography{scheduling}

\appendix

\section*{Appendix}

\section{Technical Lemmas}

\begin{lemma}
\label{lemma:tech_1}
For any nonnegative reals $x,y$ with $xy\neq 0$ and all positive reals $\beta,\gamma$:
$$
\frac{x+\beta y}{\max\sset{x,\gamma y}} \leq \frac{\beta}{\gamma}+1
$$
\end{lemma}
\begin{proof}
There are two cases to consider. First, if $x\geq \gamma y$, then
$$
\frac{x+\beta y}{\max\sset{x,\gamma y}}
=\frac{x+\beta y}{x}
=\beta\frac{y}{x}+1
\leq \beta\frac{1}{\gamma}+1.
$$
Secondly, if $\gamma y \geq x$, then
$$
\frac{x+\beta y}{\max\sset{x,\gamma y}}
=\frac{x+\beta y}{\gamma y}
=\frac{\beta}{\gamma}+\frac{x}{\gamma y}
\leq \frac{\beta}{\gamma}+1.
$$
\end{proof}

	\begin{lemma}\label{lemma:combi}
	For $i,j = 1, \ldots, n$, let reals $\alpha>0$, $a_{i,j} >
	0$ (for $i\neq j$) and $a_{i,i} =
	0$ such that for all $j$:
	$$\sum_{i=1}^n a_{i,j} < \frac{(n-1)\cdot \alpha}{\sqrt{2}}.$$
	Then, for any positive $ \varepsilon \leq \frac{\alpha}{(n-1)\sqrt{2}}$ there exists some $i$ such that:
	$$\max_{\emptyset \neq I \subseteq [n]\setminus \sset{i}} \frac{|I|}{\max_{j \in I} a_{i,j} +
		\varepsilon}
	> \frac{n-1}{\alpha \sqrt{2}}.$$
	The above inequality cannot be further improved.
\end{lemma}
\begin{proof}
	We proceed using a proof by contradiction. For any fixed $i$, we use the
	index $i_k$ to refer to $a_{i,i_k}$, the $k$-th smallest among the $a_{i,j}$.
	By
	contradiction, we have that for all $i$ and nonempty $I \subseteq [n]
	\setminus \sset{i}$:
	$$ \frac{|I|}{\max_{j \in I} a_{i,j}+\varepsilon} \le \frac{n-1}{\alpha \sqrt{2}}.$$
	In particular, for all $I_k = \sset{2,\ldots, k}$ (note that $a_{i,i_1} = a_{i,i} =
	0$) we get:
	$$
	\frac{k-1}{a_{i,{i_k}}+\varepsilon} \le
	\frac{n-1}{\alpha \sqrt{2}}
	\quad\then\quad a_{i,{i_k}} \ge \alpha\frac{(k-1)\sqrt{2}}{n-1} - \varepsilon.
	$$
	Summing over all values of $i$ and $j$:
	\begin{align*}
	\sum_{i=1}^n \sum_{j=1}^n a_{i,j}
	=\sum_{i=1}^n \sum_{k=2}^n a_{i,{i_k}}
	&\ge \sum_{i=1}^n \sum_{k=2}^n \alpha\frac{(k-1)\sqrt{2}}{n-1}-\varepsilon
	\\
	&= \frac{\alpha\sqrt{2}}{n-1} \sum_{i=1}^n \sum_{k=2}^n k-1-\varepsilon \\
	&= \frac{\alpha\sqrt{2}}{n-1} \cdot n \cdot \frac{n(n-1)}{2} -
	n(n-1)\varepsilon \\
	&= n \frac{n\cdot \alpha}{\sqrt{2}} - n(n-1)\varepsilon.
	\end{align*}
	Since the $a_{ij}$ are partitioned by the $n$ distinct values of $j$, there
	must be some $j$ for which:
	$$
	\sum_{i=1}^n a_{i,j} \ge \frac{n\cdot \alpha}{\sqrt{2}} - (n-1)\varepsilon \ge
	\frac{(n-1)\cdot \alpha}{\sqrt{2}},
	$$
	leading to a contradiction.

	\bigskip
	This result is essentially tight. For any $\delta > 0$ consider the
	matrix:
	$$
	\alpha \frac{\sqrt{2} - \delta}{n-1} \cdot
	\begin{pmatrix}
	0 & 1 & 2 & \cdots & n-2 & n-1 \\
	n-1 & 0 & 1 & \cdots & n-3 & n-2 \\
	\vdots & \vdots & \vdots & \vdots & \vdots & \vdots\\
	1  & 2 & 3  & \cdots & n-1 & 0  \\
	\end{pmatrix},
	$$
	where every column sum is $< (n-1)\alpha/\sqrt{2}$ and for every row $i$
	and nonempty $I \subseteq [n] \setminus \sset{i}$:
	$$
	\frac{|I|}{\max_{j \in I} a_{i,j}} \le \frac{|I|}{a_{i,{i_{|I|+1}}}}
	= \frac{|I|}{\alpha(\sqrt{2} -\delta)\frac{|I|}{n-1}}
	= \frac{n-1}{\alpha(\sqrt{2} - \delta)}.
	$$
\end{proof}

\section{Proof of \texorpdfstring{\cref{thm:sqrt-tradeoff}}{Theorem~9}}
\label{sec:app-sqrt2}
	\newcommand{\h}{\vecc{h}}
	Without loss of generality, we assume that mechanism $\mech$ allocates
	each
	task independently by running the \emph{same} single-task mechanism for
	\emph{every} task. The reason the analysis carries over is that we will only
	use tasks drawn from a finite pool. Restricted to these profiles of true 
	processing times,
	there are only finitely many \emph{essentially different} single-task
	mechanisms,
	when the difference is measured from the perspective of allocations.
	Therefore,
	even if the task-independent mechanism $\mech$ used a different mechanism 
	for
	every
	task, we could always find $n$ single-component mechanisms operating the
	same way. For a rigorous treatment, we refer the reader to
	\cref{lemma:multi2single}.

	We aim to identify some ``weakness'' of
	$\mech$ by discovering, for canonical cost vectors that are permutations of 
	the true
	processing times $(1,x,
	\infty, \infty, \ldots)$, just how much larger $x$ can get, so that there exists
	some equilibrium under which the  task is allocated to the slow machine. We
	refer to the machine with processing time $1$ as the \emph{fast} machine and to
	the other, which does not have a processing time of $\infty$, as the \emph{slow}
	machine. As the mechanism is \emph{not} anonymous, there might be a different
	$x$ for every permutation.

	Formally, fix $\frac{\alpha}{(n-1)\sqrt{2}} \ge \varepsilon > 0$ and for $i,j =
	1,\ldots,n$
	and $i \neq j$, let
	$$
	a_{i,j} = \max_{k \in \mathbb{N}} \sset{k\cdot \varepsilon \fwh{\text{for }
			{\t} = (\ldots, \underbrace{1}_i, \ldots,
			\underbrace{k\cdot \varepsilon}_j, \ldots),\:
			\exists {\ss}_{i,j}
			\text{ s.t. machine } j \text{ is allocated the task}}}
	$$
	be the maximum processing time of the slow machine, when the fast machine has
	index $i$, the slow machine has index $j$ and the slow machine can still receive
	the task for some equilibrium. To ensure that \cref{lemma:multi2single} can be
	applied, we need to discretize the processing times, in order to have a finite
	number of possible tasks. Let $\bar{a}_{i,j} = a_{i,j} + \varepsilon$. Clearly,
	if $a_{i,j}$ is replaced by $\bar{a}_{i,j}$ in its canonical cost vector then
	the \emph{only allocation} remaining will be to give the task to the fast
	machine. For convenience, we also set $a_{i,i} =
	\bar{a}_{i,i} = 0$. Let $\h_{i,j}$ be the permutation of the canonical cost vector 
	where 
	$i$ is the fast machine and $j$ the slow machine with cost $a_{i,j}$ and 
	$\bar{\h}_{i,j}$ the same but with cost $\bar{a}_{i,j}$ instead.

	Next, observe that $0 < a_{i,j} < (n-1)\alpha/\sqrt{2}+1$ for all $i \neq j$. The 
	upper
	bound
	is due to the assumption on the $\poa$, which would otherwise be violated for
	${\h}_{i,j}$. For the lower bound, the instance $\bar{{\h}}_{i,j}$
	would have $\poa
	\ge 1/\varepsilon$, leading to a contradiction, for some small enough
	$\varepsilon$.

	\bigskip
	For every machine $j$, we create the following profile of true processing 
	times:
	$$
	{\t}(j)
	= ( {\h}_{1,j}, \ldots , {\h}_{(i-1),j},  \vecc{f}^\star(j),
	{\h}_{(i+1),j}, \ldots, {\h}_{n,j})^\top,
	$$
	where $\vecc{f}^\star$ contains only one entry that is not $\infty$, at row $j$.
	Presented in a more convenient matrix form:
	$$
	\left.\begin{pmatrix}
	1 & \infty & \infty & \infty & \infty &\cdots & \infty & \infty & \infty \\
	\infty & 1 & \infty & \infty & \infty &\cdots & \infty  & \infty & \infty \\
	\vdots   & \ddots & \ddots &  \ddots & \ddots &\ddots & \ddots &  \ddots &
	\infty\\
	\infty  & \infty  & \cdots & 1 & \infty &\infty & \cdots  & \infty & \infty\\
	a_{1,j}  & a_{2,j}  & \cdots & a_{(j-1),j} & 1 & a_{(j+1),j} & \cdots & a_{(n-1),j} &
	a_{n,j} \\
	\infty  & \infty  & \cdots & \infty & \infty &1 & \infty & \cdots & \infty\\
	\vdots   & \vdots & \vdots & \vdots & \vdots &\ddots & \ddots & \ddots &
	\vdots\\
	\infty & \infty & \cdots & \infty & \infty &\infty & \infty & 1 & \infty \\
	\infty & \infty & \cdots & \infty & \infty &\infty & \infty & \infty & 1
	\end{pmatrix}\quad\right\}n,
	$$
	where machine $j$ is the slow machine for all the tasks and for each task there exists exactly one,
	distinct, fast machine. By task-independence, the profile of reported processing times
	$$( {\ss}_{1,j}, {\ss}_{2,j}, \ldots ,
	{\ss}_{(i-1),j},\vecc{f}^\star_{EQ}, {\ss}_{(i+1),j},
	\ldots, {\ss}_{n,j})^\top$$ is an equilibrium in which all tasks are allocated to
	the slow machine and $\vecc{f}^\star_{EQ}$ is any equilibrium for task
	$j$, which clearly has to select the $j$-th machine.  Since
	the optimal makespan is $1$, by the given bound on the $\poa$,
	we have that for all $j$:
	\begin{equation}\label{eq:poa_upper_bound}
		\sum_{i=1}^n a_{i,j} + 1 < \frac{(n-1)\cdot\alpha}{\sqrt{2}} + 1 \quad\then\quad
		\sum_{i=1}^n a_{i,j} < \frac{(n-1)\cdot\alpha}{\sqrt{2}}
	\end{equation}

	In the proof of~\cref{thm:task_independent_tradeoff}, in order to obtain the
	bound for the $\pos$, we swapped the positions of the slow and fast machines on
	each column, which was facilitated by the assumption that the mechanism in that
	case was anonymous. In this case however mechanism $\mech$ is not anonymous, so
	we have to make a slight adjustment and change $a_{i,j}$ to $\bar{a}_{i,j}$. To
	obtain a lower bound on the $\pos$, we instead fix $i$ and using the true
	processing times $$\bar{\t}(i) = ( \bar{\h}_{i,1}, \ldots , \bar{\h}_{i,(i-1)},
	\bar{\vecc{f}}^\star,
	\bar{\h}_{i,(i+1)}, \ldots, \bar{\h}_{i,n})^\top,$$
	we get a matrix similar to the previous one, i.e.,
	$$
	\begin{pmatrix}
	\bar{a}_{i,1} & \infty & \infty & \infty & \infty &\cdots & \infty & \infty & \infty \\
	\infty & \bar{a}_{i,2} & \infty & \infty & \infty &\cdots & \infty  & \infty & \infty \\
	\vdots   & \ddots & \ddots &  \ddots & \ddots &\ddots & \ddots &  \ddots &
	\infty\\
	\infty  & \infty  & \cdots & \bar{a}_{i,(i-1)} & \infty &\infty & \cdots  & \infty &
	\infty\\
	1  & 1  & \cdots & 1 & \max_j \bar{a}_{i,j} & 1 & \cdots & 1 & 1 \\
	\infty  & \infty  & \cdots & \infty & \infty &\bar{a}_{i,(i+1)} & \infty & \cdots &
	\infty\\
	\vdots   & \vdots & \vdots & \vdots & \vdots &\ddots & \ddots & \ddots &
	\vdots\\
	\infty & \infty & \cdots & \infty & \infty &\infty & \infty & \bar{a}_{i,(n-1)} & \infty
	\\
	\infty & \infty & \cdots & \infty & \infty &\infty & \infty & \infty & \bar{a}_{i,n}
	\end{pmatrix}
	$$

	Again, by task independence and the definition of $\bar{a}_{i,j}$, the mechanism
	has to allocate each task to the \emph{fast} machine, which for all tasks is
	machine $i$. The optimal allocation could potentially only allocate to slow
	machines and produce an allocation with makespan at most $\max_{j}
	\bar{a}_{i,j}$. The lower bound on the $\pos$ can be further improved by
	restricting the input to only contain a nonempty subset $I \subseteq [n]
	\setminus \sset{i}$ of the tasks, plus task $i$ which is special. This reduces
	the makespan of \emph{both} $\mech$ and the optimal, but the overall fraction
	could increase. Notice that task $i$ (which has only one entry which is not
	$\infty$, adjusted to be the maximum amongst the $|I|$ selected $\bar{a}_{i,j}$)
	is \emph{always} added, as it increases $\mech$'s makespan for free. In
	particular:
	\begin{equation}\label{eqn:pos_bound}
		\pos
		\ge
		\max_{I \subseteq [n]\setminus \sset{i}} \frac{|I| + \max_{j \in I}
			\bar{a}_{i,j}}{\max_{j \in I} \bar{a}_{i,j}}
		=
		\max_{I \subseteq [n]\setminus \sset{i}} \frac{|I|}{\max_{j \in I} \bar{a}_{i,j}} +
		1
	\end{equation}

	Notice that the $a_{i,j}$ satisfy the premises of \cref{lemma:combi}: the `if' 
	part holds because each $\t(j)$ generates one such inequality by the $\poa$ 
	bound. Similarly, any inequality generated by the `then' part can be captured 
	by an appropriate $\pos$ bound in $\bar{\t}(i)$.
	Applying the result of \cref{lemma:combi} to \cref{eqn:pos_bound} we get:
	$$
	\pos
	\ge
	\max_{I \subseteq [n]\setminus \sset{i}} \frac{|I|}{\max_{j \in I} \bar{a}_{i,j}} +
	1
	=
	\max_{I \subseteq [n]\setminus \sset{i}} \frac{|I|}{\max_{j \in I}
		a_{i,j} + \varepsilon} + 1
	> \frac{n-1}{\alpha\sqrt{2}} + 1,
	$$
	which completes the proof.

\

The proof of \cref{thm:sqrt-tradeoff} shows that foregoing anonymity cannot lead to
mechanisms with \emph{asymptotically} tighter Pareto frontiers. Moreover, it leads to the
following observations. Most likely, if there exists a stronger bound, it cannot be
obtained with our canonical cost vectors and would need to use a richer input. At
the same time, if a non-anonymous mechanism that achieves a better trade-off than
that given by \cref{thm:task_independent_tradeoff} does exist, the matrix at the end
of \cref{lemma:combi} could provide some insight on what that mechanism could look
like, at least when restricted to these tasks.

\begin{lemma}\label{lemma:multi2single}
	\cref{thm:sqrt-tradeoff} holds for all task-independent mechanisms
	$\mech$.
\end{lemma}
\begin{proof}
	\newcommand{\tasks}{\mathcal{C}}
	Continuing from our note in the beginning of \cref{thm:sqrt-tradeoff},
	we show that given a task-independent mechanism $\mech$, it is always
	possible to find $n$ single-task component
	mechanisms that
	behave the same way for the very specific inputs needed for the lower
	bounds.
	Notice that in the proof of \cref{thm:sqrt-tradeoff} the tasks we used
	all come from a finite set of true processing times $\tasks$. In particular,
	$\tasks$ contains all the
	permutations of $(1, a, \infty, \infty, \ldots, \infty)$, where $a$ is of the form
	$k\cdot \varepsilon$ and $0 < a < (n-1)\cdot \alpha /\sqrt{2} + \varepsilon$.
	Therefore:
	$$
	\left|\tasks\right| = n(n-1) \frac{(n-1)\cdot \alpha /\sqrt{2}}{\varepsilon}.
	$$
	Fix a single-task mechanism $\mathcal{A}$.  For a true processing time vector
	${\t} \in \tasks$, call machine $j$ \emph{accessible} from ${\t}$ if there
	exists an equilibrium for which the task is assigned to that machine.
	Let $\mathit{ACC}({\t})$ be the set of accessible machines for component
	$\mathcal{A}$ for that ${\t}$ and $\mathcal{B}(\mathcal{A})$ be the
	\emph{behaviour} set of $\mathcal{A}$, defined as
	$$
	\mathcal{B}(\mathcal{A}) = \sset{({\t}, \mathit{ACC}({\t}))
		\fwh{{\t} \in \tasks}}.
	$$
	This fully characterizes $\mathcal{A}$ from the perspective of
	accessible allocations for tasks in $\tasks$. Note that for two different
	single-task mechanisms $\mathcal{A}_1$ and $\mathcal{A}_2$, it could be that
	$\mathcal{B}(\mathcal{A}_1) = \mathcal{B}(\mathcal{A}_2)$ without them being
	the same mechanism: they must however reach the same allocations given
	the same ${\t}$ (potentially through different equilibria).

	Clearly, $\mathit{ACC}({\t})$ can be one of $2^n-1$ possible sets. This is
	because
	$\mathit{ACC}({\t})$ can be any subset of the $n$ machines, except the empty
	subset, since each task has to be allocated
	to some machine at every equilibrium. This is however all we need in
	\cref{thm:sqrt-tradeoff}. Therefore, the total number of behaviour
	sets is at most:
	$$
	(2^{n}-1)^{|\tasks|},
	$$
	which is quite large, but bounded.

	Let $\mathcal{A}_j$ be the single-task mechanism used by mechanism 
	$\mech$ for the
	$j$'th
	task, given an input with $n(2^{n}-1)^{|\tasks|}$ tasks from $\tasks$ in total.
	By the pidgeonhole principle, there must be at least $n$ such single-task
	mechanisms with the same behaviour set. Setting all other tasks to have real
	processing times $0$ for all tasks, we have successfully extracted a
	mechanism $\mech'$ that behaves exactly as needed for
	\cref{thm:sqrt-tradeoff}.
\end{proof}

\end{document}